\documentclass[11pt]{article}

\usepackage{bm, commath}
\usepackage{natbib}
\usepackage{caption}
\usepackage{graphicx}
\usepackage{subcaption}
\usepackage{amsmath, amsfonts, amsthm}
\usepackage{float}
\usepackage{booktabs,siunitx}
\usepackage{url}
\usepackage{multirow}
\usepackage{amssymb}
\usepackage{bm}
\usepackage{mathrsfs}
\usepackage{xr}
\usepackage{authblk}
\usepackage[table]{xcolor}

\usepackage{listings}
\allowdisplaybreaks
\usepackage{bbm}
\usepackage{textcomp}
\usepackage[margin=1in]{geometry}
\usepackage{authblk}
\usepackage[ruled]{algorithm2e}
\SetKwInput{KwParam}{Parameter}
\SetAlgoCaptionLayout{centerline}

\usepackage{sectsty}
\usepackage[onehalfspacing]{setspace}
\usepackage[utf8]{inputenc}
\usepackage{float}
\usepackage{tikz, pgfplots}
\usepackage{booktabs}
\usetikzlibrary{intersections,decorations.pathreplacing,arrows.meta,calc,fit,patterns,shapes,shapes.misc,shapes.geometric,positioning,angles,quotes}
\usepgfplotslibrary{fillbetween}

\definecolor{dblue}{HTML}{0072B2}
\definecolor{dorange}{HTML}{D55E00}
\definecolor{dlightblue}{HTML}{97C5DF}
\definecolor{dlightorange}{HTML}{F1CAAC}
\definecolor{dgreen}{rgb}{0.,0.6,0.}

\tikzset{
	hatch distance/.store in=\hatchdistance,
	hatch distance=10pt,
	hatch thickness/.store in=\hatchthickness,
	hatch thickness=2pt
}

\tikzset{
	invisible/.style={opacity=0},
	visible on/.style={alt={#1{}{invisible}}},
	alt/.code args={<#1>#2#3}{%
		\alt<#1>{\pgfkeysalso{#2}}{\pgfkeysalso{#3}} 
	},
}

\newsavebox\mybox

\def\signed #1{{\leavevmode\unskip\nobreak\hfil\penalty50\hskip2em
		\hbox{}\nobreak\hfil(#1)%
		\parfillskip=0pt \finalhyphendemerits=0 \endgraf}}

\makeatletter
\pgfdeclarepatternformonly[\hatchdistance,\hatchthickness]{flexible hatch}
{\pgfqpoint{0pt}{0pt}}
{\pgfqpoint{\hatchdistance}{\hatchdistance}}
{\pgfpoint{\hatchdistance-1pt}{\hatchdistance-1pt}}%
{
    \pgfsetlinewidth{\hatchthickness}
    \pgfpathmoveto{\pgfqpoint{0pt}{0pt}}
    \pgfpathlineto{\pgfqpoint{\hatchdistance}{\hatchdistance}}
    \pgfusepath{stroke}
}
\makeatother

\newtheorem{theorem}{Theorem}

\newtheorem{assumption}{Assumption}

\newtheorem{lemma}{Lemma}

\usepackage{mathtools}

\makeatletter
\renewcommand{\algocf@captiontext}[2]{#1\algocf@typo. \AlCapFnt{}#2} 
\def\@algocf@capt@plain{top}
\renewcommand{\algocf@makecaption}[2]{%
  \addtolength{\hsize}{\algomargin}%
  \sbox\@tempboxa{\algocf@captiontext{#1}{#2}}%
  \ifdim\wd\@tempboxa >\hsize
  \hskip .5\algomargin%
  \parbox[t]{\hsize}{\algocf@captiontext{#1}{#2}}
  \else%
  \global\@minipagefalse%
  \hbox to\hsize{\box\@tempboxa}
  \fi%
  \addtolength{\hsize}{-\algomargin}%
}
\makeatother

\def\var{{\rm Var}}
\def\cov{{\rm Cov}}
\def\corr{{\rm Corr}}
\def\bfX{\boldsymbol{X}}
\def\bbeta{\boldsymbol{\beta}}

\begin{document}

\sectionfont{\bfseries\large\sffamily}%

\subsectionfont{\bfseries\sffamily\normalsize}%

\title{Behavioral Carry-Over Effect and Power Consideration in Crossover Trials}

\author{Danni Shi and Ting Ye\thanks{tingye1@uw.edu. This work was supported by the HIV Prevention Trials Network (HPTN) and NIH grant: NIAID 5 UM1 AI068617.}}
\affil{Department of Biostatistics, University of Washington,  Seattle, Washington 98195, U.S.A.}


\maketitle

\begin{abstract}
 A crossover {trial} is an efficient trial design when there is no carry-over effect. To reduce the impact of the biological carry-over effect, a washout period is often designed. However, the carry-over effect remains an outstanding concern when a washout period is unethical or cannot sufficiently diminish the impact of the carry-over effect. The latter can occur in comparative effectiveness research where the carry-over effect is often non-biological but behavioral. In this paper, we investigate the crossover design under a potential outcomes framework with and without the carry-over effect. We find that when the carry-over effect exists and satisfies a sign condition, the basic estimator underestimates the treatment effect, which does not inflate the type I error of one-sided tests but negatively impacts the power. This leads to a power trade-off between the crossover design and the parallel-group design, and we derive the condition under which the crossover design does not lead to type I error inflation and is still more powerful than the parallel-group design. We {also} develop covariate adjustment methods for crossover trials. We evaluate the performance of cross-over design and covariate adjustment using data from the MTN-034/REACH study.

\end{abstract}

\section{Introduction}
\label{sec: intro}
A crossover trial is a longitudinal study in which patients are randomized to different sequences of treatments with the intention of comparing the effects of different treatments \citep[Chapter 8]{hills1979two,senn1994ab,kim2021handbook}. In the simplest crossover design, patients are randomly allocated to two groups: one group receives treatment 1 at period 1 and treatment 0 at period 2, and the other group receives the treatment sequence in the reverse order. The main advantage of a crossover design is that every patient can serve as their own control, which largely eliminates the between-patient variability and {makes the crossover design more efficient} than the parallel-group design. 

A key assumption for the crossover design is that there is no \textit{carry-over effect}: that is, the treatment assigned during the first period does not interfere with the outcome in the second period. Therefore, a crossover design is most common for treatments whose effect vanishes when discontinued and for non-absorbing endpoints. Examples include early phase trials (e.g., pharmacokinetic (PK) studies and dose-finding studies) and phase III studies with chronic conditions (e.g., hypertension, pain, and asthma); see \cite{jones1995case} for a review. 
Meanwhile, a washout period is often designed between period 1 and period 2 to effectively reduce the impact of the \textit{biological carry-over effect} engendered by treatments taken at period 1. {The washout can be \emph{passive} (i.e., no treatment is given during the washout period) or \emph{active} 
	(i.e., a treatment is given during the washout period but measurement is delayed until steady state is reached) \citep{senn2002time, araujo2016washout}.} However, the carry-over effect remains an outstanding concern when a washout period is inappropriate or cannot sufficiently diminish the impact of the first period. {For example, the MTN-034/REACH study \citep{nair2023adherence} is a phase 2a HIV-prevention trial evaluating the safety of and the adherence to the monthly dapivirine vaginal ring (DVR) and daily oral tenofovir disoproxil fumarate plus emtricitabine (TDF/FTC) among adolescent girls and young women aged 16-21 years. This trial uses a two-period crossover design without a passive washout period since withholding an effective HIV preventive agent from the population at risk is unethical. 
	{If participants initially using the monthly DVR (which requires no active effort to remain adherent for a month after inserting the ring) find the transition to daily pill-taking in period 2 particularly burdensome, then the ease and habit formed during the first period with the monthly DVR might carry over into period 2 and affect their adherence behavior in period 2.} This type of carry-over effect is non-biological and is hard to be eliminated by washout periods; we term it as the \textit{behavioral carry-over effect}, which is the impact that a treatment has on the subsequent outcomes {due to the altering of a participant's behavior.} It can be an important consideration in other HIV prevention trials, e.g., the TRIO study  \citep{minnis2018young}, and more broadly in comparative effectiveness research that uses the crossover design \citep{hemming2020use}.}

The presence of the carry-over effect can bias the estimation of the treatment effect. This motivates \cite{grizzle1965two}'s two-stage procedure, which first tests whether the carry-over effect exists or not, and the testing result determines whether to use the period 2 data in the analysis. This two-stage procedure was then criticized by \cite{freeman1989performance} because it has low power and can inflate the type I error of subsequent analysis. Another type of method is to model the carry-over effect \citep{brown1980crossover,laird1992analysis, jones1996modelling, kunert2002optimal,bailey2006optimal}, which are sensitive to model misspecification. 

In this paper, we examine the carry-over effect in crossover trials using a potential outcomes framework \citep{Neyman:1923a, Rubin:1974} and only require minimal statistical assumptions to gain a better understanding of its impact. Under the ideal setting with no carry-over effect, we derive the asymptotic properties of the basic estimator of the treatment effect, based on which we show that the crossover design is typically more efficient than the parallel-group design. Our findings, while aligning with existing efficiency comparisons in the literature (refer to \citet[Section 9.2.1]{senn2002time} for example), are more general, because we only assume the no carry-over effect assumption and no other assumptions (e.g., modeling assumptions or normality of data). Then we investigate the situation with the carry-over effect, uncovering new insights on its influence on estimation bias, as well as the type I error and power of tests. When the carry-over effect $\lambda_1+\lambda_0$ is negative, the basic estimator overestimates the treatment effect, which can inflate the type I error of one-sided tests. Conversely, when the carry-over effect is positive, the basic estimator underestimates the treatment effect, which does not inflate the type I error of one-sided tests but negatively impacts the power.  This leads to a power trade-off between the crossover design and the parallel-group design, and we derive the condition under which the crossover design does not lead to type I error inflation and is still more powerful than the parallel-group design. 

The rest of the paper proceeds as follows. Section \ref{sec: no carry-over} focuses on crossover trials with no carry-over effect, where we present a potential outcomes framework for crossover trials, investigate the basic estimator, and compare the efficiencies of a crossover and a parallel-group design. Section \ref{sec: crossover} extends the results to the case with the carry-over effect. Section \ref{sec: covariate adjustment} proposes a covariate adjustment method and studies its asymptotic properties. Section \ref{sec: power} examines the power trade-off of the crossover and the parallel-group design. Section \ref{sec: mtn034} presents the application of our methods to the REACH study, as well as power considerations in the design of a future HIV prevention trial using data from the REACH study.
Section \ref{sec: discussion} concludes with a discussion.

\vspace{-5mm}
\section{Crossover trials with no carry-over effect}
\label{sec: no carry-over}
\subsection{Setup and assumptions}
Consider a crossover trial for two treatments in two periods. A sample of $n$ subjects are randomly allocated to two treatment sequences, where $A_i=1$ denotes that subject $i$ first receives treatment 1 and then treatment 0, and $A_i=0$ denotes the reverse order. Let $Y_{i1}^{(j)} $ be the potential outcome at time 1 had the subject been exposed to treatment $j$ at time 1, for $j=0,1$. Let $Y_{i2}^{(jk)} $ be the potential outcome at time 2 had the subject been exposed to treatment $j$ at time 1 and treatment $k$ at time 2, for $j, k=0,1$. The observed outcome for subject $i$ at time $t$ is $Y_{it}$. Throughout the article, we make the consistency assumption that links the observed outcomes to the potential outcomes: for $A_{i}=1, Y_{i1} = Y_{i1}^{(1)}$ and $ Y_{i2}= Y_{i2}^{(10)}; $ for $A_{i}=0, Y_{i1} = Y_{i1}^{(0)}$ and $Y_{i2}= Y_{i2}^{(01)}$. Let $\bfX_i$ be a vector of observed baseline covariates for subject $i$. We assume that $ \big(A_i, \bfX_i, Y_{i1}^{(j)}, Y_{i2}^{(jk)}, j,k=0,1 \big), i=1,\dots, n$ are independent and identically distributed with finite second-order moments, and the covariance matrix $\var (\bfX_i) $ is positive definite.

Simple randomization assigns subjects to the two treatment sequences completely at random. This is summarized in Assumption \ref{assump: random}.

\begin{assumption}[Simple randomization] \label{assump: random}
	$A_i \perp \big(\bfX_i, Y_{i1}^{(j)}, Y_{i2}^{(jk)} \big)$ for $j, k= 0,1$, and $P(A_i=1)=\pi_1$ where $0<\pi_1<1$ is known and $\pi_0=1-\pi_1$.
\end{assumption}

Assumptions \ref{assump: no carry over} is the key assumption typically imposed in crossover trials: it states
that the treatment at time 1 does not have a direct effect on the outcome at time 2; see Figure \ref{fig: all assump}(a) for an illustration. Many crossover trials would plan a sufficiently long washout period between the two time periods to make this assumption more plausible. 

\begin{assumption}[No carry-over effect] \label{assump: no carry over}
	For $k =0,1$, $Y_{i2}^{(0k)} = Y_{i2}^{(1 k)}:= Y_{i2}^{(k)}$ almost surely. 
\end{assumption}

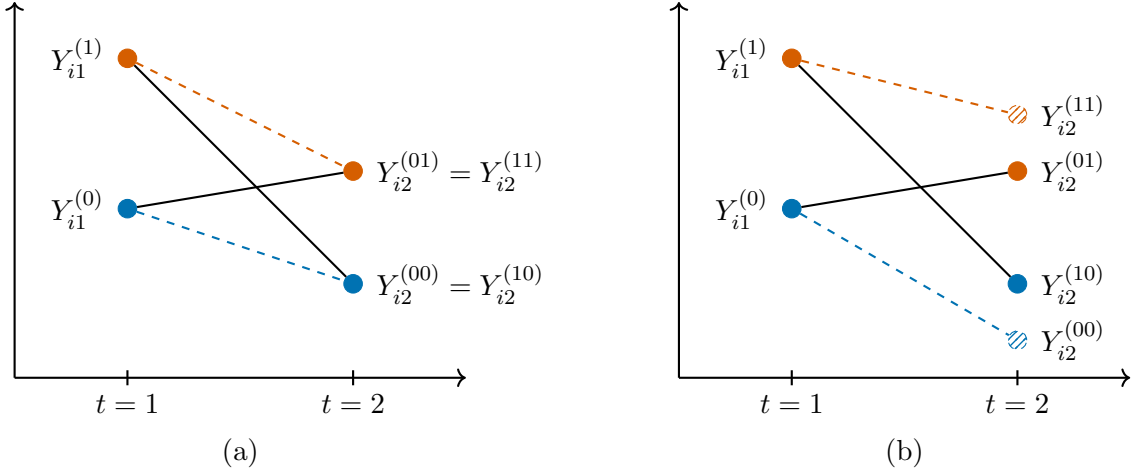
\begin{figure}[h]
	\begin{subfigure}[t]{.5\textwidth}
		\centering
		\resizebox{!}{!}{%
			\begin{tikzpicture} 
				\tikzstyle{dp}=[circle, inner sep=0pt, minimum size=7pt]
				\tikzstyle{trlink}=[draw=black, thick]
				
				\tikzstyle{uc}=[dp, draw=dblue, fill=dblue]
				\tikzstyle{uclink}=[draw=dblue, thick]
				
				\tikzstyle{lc}=[dp, draw=dorange, fill=dorange]
				\tikzstyle{lclink}=[draw=dorange, thick]
				
				\draw [->, thick] (0,0) -- (0,5) {}; 
				\draw [->, thick] (0,0) -- (6,0) {};
				\draw[thick,-] (1.5,-0.1) -- (1.5,0.1) node[anchor=north,below=5pt] {$t=1$};
				\draw[thick,-] (4.5,-0.1) -- (4.5,0.1) node[anchor=north,below=5pt] {$t=2$};
				\node at (3,-1) {(a)};
				
				\node[uc] (ua) at (1.5,2.25) {};
				\node[left=5pt] at (ua) {$Y_{i1}^{(0)}$};
				\node[lc] (la) at (1.5,4.25) {};
				\node[left=5pt] at (la) {$Y_{i1}^{(1)}$};
				\node[uc] (uc) at (4.5,1.75-0.5) {};
				\node[right=5pt] at (uc) {$Y_{i2}^{(00)}= Y_{i2}^{(10)}$};
				
				\node[lc] (lc) at (4.5,3.75-1) {};
				\node[right=5pt] at (lc) {$Y_{i2}^{(01)} = Y_{i2}^{(11)} $};
				
				\draw [uclink,dashed] (ua) -- (uc);
				\draw [lclink,dashed] (la) -- (lc);
				
				\draw [trlink] (ua) -- (lc);
				\draw [trlink] (uc) -- (la);
		\end{tikzpicture}}	
	\end{subfigure}%
	{\begin{subfigure}[t]{.5\textwidth}
			\centering
			\resizebox{!}{!}{%
				\begin{tikzpicture} 
					\tikzstyle{dp}=[circle, inner sep=0pt, minimum size=7pt]
					\tikzstyle{trlink}=[draw=black, thick]
					
					\tikzstyle{uc}=[dp, draw=dblue, fill=dblue]
					\tikzstyle{ucfac}=[dp, draw=dblue, dashed]
					\tikzstyle{uclink}=[draw=dblue, thick]
					
					\tikzstyle{lc}=[dp, draw=dorange, fill=dorange]
					\tikzstyle{lcfac}=[dp, draw=dorange, dashed]
					\tikzstyle{lclink}=[draw=dorange, thick]
					
					\draw [->, thick] (0,0) -- (0,5) {}; 
					\draw [->, thick] (0,0) -- (6,0) {};
					\draw[thick,-] (1.5,-0.1) -- (1.5,0.1) node[anchor=north,below=5pt] {$t=1$};
					\draw[thick,-] (4.5,-0.1) -- (4.5,0.1) node[anchor=north,below=5pt] {$t=2$};
					\node at (3,-1) {(b)};
					
					\node[uc] (ua) at (1.5,2.25) {};
					\node[left=5pt] at (ua) {$Y_{i1}^{(0)}$};
					\node[lc] (la) at (1.5,4.25) {};
					\node[left=5pt] at (la) {$Y_{i1}^{(1)}$};
					\node[uc] (uc) at (4.5,1.75-0.5) {};
					\node[right=5pt] at (uc) {$Y_{i2}^{(10)}$};
					\node[ucfac, pattern=flexible hatch, hatch distance=4pt, hatch thickness=.5pt, pattern color=dblue] (ud) at (4.5,1.75-1.25) {};
					\node[right=5pt] at (ud) {$Y_{i2}^{(00)}$};
					
					\node[lc] (lc) at (4.5,3.75-1) {};
					\node[right=5pt] at (lc) {$Y_{i2}^{(01)} $};
					\node[lcfac, pattern=flexible hatch, hatch distance=4pt, hatch thickness=.5pt, pattern color=dorange] (ld) at (4.5,3.75-0.25) {};
					\node[right=5pt] at (ld) {$Y_{i2}^{(11)} $};
					
					\draw [uclink,dashed] (ua) -- (ud);
					\draw [lclink,dashed] (la) -- (ld);
					
					\draw [trlink] (ua) -- (lc);
					\draw [trlink] (uc) -- (la);
			\end{tikzpicture}}
	\end{subfigure}}
	\caption{(a) Six potential outcomes under Assumption \ref{assump: no carry over}. (b) Six potential outcomes without Assumption \ref{assump: no carry over}.}
	\label{fig: all assump}
\end{figure}

Under Assumption \ref{assump: no carry over}, we can simply let $Y_{i2}^{(k)}, k=0,1$ denote the potential outcome at time 2. We are interested in the average treatment effects {$2^{-1} (\theta_1 + \theta_2)$}, where $\theta_1= E\big(Y_{i1}^{(1)} - Y_{i1}^{(0)}\big)$ and $\theta_2= E\big( Y_{i2}^{(1)} -Y_{i2}^{(0)} \big)$. Note that it is common to assume that the treatment effect is time-invariant so that $\theta_1 = \theta_2$ \citep{hills1979two}. However, as shown in Theorem \ref{theo: basic estimator} below, a time-invariant treatment effect assumption is not necessary for studying treatment effect in crossover trials. 

\subsection{The basic estimator}
Using the full crossover period, we can calculate the difference in outcomes for every subject under treatment 1 and treatment 0. The \textit{basic estimator} is commonly used, which first calculates the arm-specific average outcome difference, and then averages the two means from the two arms \citep[Chapter 8]{hills1979two,senn1994ab,kim2021handbook}:
\begin{align}\label{eq: basic estimator}
	\hat \theta_{\rm cr}=\frac12 \bigg\{  \frac{1}{n_1} \sum_{i=1}^{n}  A_i  ( Y_{i1} - Y_{i2} )  +  \frac{1}{n_0} \sum_{i=1}^{n} (1-A_i)  (  Y_{i2}  - Y_{i1} )\bigg\}:= \frac{\bar \Delta_1 - \bar \Delta_0 }{2},
\end{align}
where $\bar \Delta_a = n_a^{-1} \sum_{A_i =a} \Delta_i $ is the average of $\Delta_i$'s for subjects with $A_i=a$, $\Delta_i = Y_{i1} - Y_{i2}$, and $n_a$ is the number of subjects with $A_i=a$, for $a=0,1$. Theorem \ref{theo: basic estimator} summarizes the statistical properties of the 
basic estimator $\hat \theta_{\rm cr}$. 

\begin{theorem}\label{theo: basic estimator}
	Under Assumptions \ref{assump: random}-\ref{assump: no carry over}, \\
	(a) $E(	\hat \theta_{\rm cr}) = 2^{-1}(\theta_1+\theta_2)$, where $\theta_t = E\big(Y_{it}^{(1)} - Y_{it}^{(0)} \big)$ for $t=1,2$.\\
	(b) $	\sqrt{n} \{	\hat \theta_{\rm cr} - 2^{-1}(\theta_1+\theta_2)\} \xrightarrow{d} N( 0,  \sigma^2_{\rm cr})$, where $\sigma^2_{\rm cr}  = (4\pi_1)^{-1} \var\big(Y_{i1}^{(1)}-Y_{i2}^{(0)} \big)$ $+(4\pi_0)^{-1}$ $\var\big(Y_{i1}^{(0)} - Y_{i2}^{(1)} \big)$.
\end{theorem}
Proof of Theorem \ref{theo: basic estimator} and all other proofs are in the supplementary materials. 
In the proof of Theorem \ref{theo: basic estimator}(a), we show that the expected change in outcome for $A_i=1$ is the average treatment effect at time 1 minus the expected change in outcome in the absence of the treatment, i.e., $E(Y_{i1} - Y_{i2} \mid A_i=1)= E\big(Y_{i1}^{(1)} - Y_{i2}^{(0)} \big)= E\big(Y_{i1}^{(1)} -Y_{i1}^{(0)}+Y_{i1}^{(0)}-  Y_{i2}^{(0)} \big) = \theta_1 - \tau$, where $\tau = E\big(Y_{i2}^{(0)} - Y_{i1}^{(0)} \big)$ denotes the expected {temporal} change in outcome in the absence of the treatment. Similarly, the expected change in outcome for $A_i=0$ is the average treatment effect plus the expected change in outcome in the absence of the treatment, i.e., $E( Y_{i2} - Y_{i1} \mid A_i=0)= E\big( Y_{i2}^{(1)} - Y_{i1}^{(0)} \big)= E\big(Y_{i2}^{(1)} - Y_{i2}^{(0)} + Y_{i2}^{(0)}-Y_{i1}^{(0)} \big) = \theta_2 + \tau$. At first glance, it appears that both the group-specific average changes in outcome are biased by $\tau$. However, randomization balances out the effect of temporal trend that is not due to the treatment, and thus the overall average change in outcome remains unbiased for $2^{-1} (\theta_1+\theta_2)$. Moreover, with a given sample size $n$, Theorem \ref{theo: basic estimator}(b) guides the optimal choice of $\pi_1$ to minimize the asymptotic variance. For example, if $\var\big(Y_{i1}^{(1)}-Y_{i2}^{(0)} \big)=\var \big(Y_{i1}^{(0)} - Y_{i2}^{(1)} \big)$, the optimal choice would be equal allocation with $\pi_1=\pi_0= 1/2$. For statistical inference, $\sigma^2_{\rm cr}$ can be consistently estimated by 
\begin{equation}\label{eq: basic est var}
	\hat \sigma_{\rm cr}^2 = \frac{1}{4\pi_1} S_{\Delta1}^2 + \frac{1}{4\pi_0}  S_{\Delta0}^2, 
\end{equation}
where $S_{\Delta1}^2 $ and $S_{\Delta0}^2 $  are respectively the sample variance of $\Delta_i = Y_{i1}- Y_{i2}$ for subjects under $A_i=1$ and $A_i=0$. 

There is another estimator that looks similar to the basic crossover estimator $\hat\theta_{\rm cr}$: $\hat\theta_{\rm cr}^{\rm alt} = n^{-1}\sum_{i=1}^{n} \{ A_i  ( Y_{i1} - Y_{i2} )  +  (1-A_i)  (  Y_{i2}  - Y_{i1} ) \}$.  These two estimators $\hat \theta_{\rm cr}$ and  $\hat \theta_{\rm cr}^{\rm alt} $ are the same under randomization schemes that enforce $n_1=n_0$, but they are not the same in other cases (e.g., under the simple randomization considered in this article). Specifically, under simple randomization with equal allocation ($\pi_1=\pi_0=1/2$), $\hat \theta_{\rm cr}^{\rm alt} $ is also unbiased for $2^{-1} (\theta_1+\theta_2)$; however, the asymptotic variance of $\sqrt{n} \big\{\hat \theta_{\rm cr}^{\rm alt}  - 2^{-1}(\theta_1+\theta_2) \big\}$ equals $2^{-1} \var \big(Y_{i1}^{(1)}-Y_{i2}^{(0)} \big)  + 2^{-1}\var \big(Y_{i1}^{(0)} - Y_{i2}^{(1)} \big) +4^{-1}(\theta_1- \theta_2 - 2\tau)^2$, which is larger than the asymptotic variance of $\hat \theta_{\rm cr}$. The additional variance component  $4^{-1}(\theta_1- \theta_2 - 2\tau)^2$ is due to the random effect of the temporal trend and the treatment effect heterogeneity at the two time points. Under simple randomization with unequal allocation, $\hat \theta_{\rm cr}^{\rm alt} $ is biased due to the time effect $\tau$. This important point is also discussed in \cite{hills1979two}. Therefore, we do not consider this alternative estimator $\hat\theta_{\rm cr}^{\rm alt} $ in the rest of the article.

\subsection{Efficiency comparison between a crossover and parallel-group design}
When there is no carry-over effect, i.e., when Assumption \ref{assump: no carry over} holds, it is well known that a crossover design is typically more efficient than a parallel-group one. This is because each subject can serve as their own control, which cuts down half the sample size, and the within-subject comparison can further remove the inter-subject variability \citep{jones2014design}. Prior work such as \citet[Section 9.2.1]{senn2002time} has provided efficiency comparisons between a crossover and parallel-group design when there is no carry-over effect and under various additional assumptions (e.g., modeling assumptions, time-invariant treatment effect, and normality of data). In this section, we provide a more general efficiency comparison which assumes only Assumption \ref{assump: no carry over} and no other assumptions.

Suppose a randomized controlled trial is being planned to demonstrate the superiority or non-inferiority of an investigational treatment. First, consider a crossover design with the null hypothesis $H_0:  2^{-1}(\theta_1+\theta_2) = \theta^*$ versus the alternative hypothesis $H_A:  2^{-1}(\theta_1+\theta_2)> \theta^*$ for some pre-specified $\theta^*$, where $\theta^*=0$ for test of superiority and $\theta^*>0$ for test of non-inferiority. The test statistic based on the basic estimator is $T_{\mathrm{cr}} =\sqrt{n} \big(\hat{\theta}_{\mathrm{cr}} -\theta^* \big)/ \hat{\sigma}_{\mathrm{cr}}$. 
From Theorem \ref{theo: basic estimator}, $T_{\mathrm{cr}} \xrightarrow{d} N(0,1)$ under $H_0$, and thus, we reject $H_0$ if and only if $T_{\mathrm{cr}} >z_{1-\alpha}$, where $\alpha$ is the significance level and $z_{1-\alpha}$ is the $(1-\alpha)$th quantile of the standard normal distribution.
Under the local alternative $\sqrt{n} \big\{2^{-1}(\theta_1+\theta_2) -\theta^* \big\} \rightarrow \gamma_{\rm cr}$ with a constant $\gamma_{\rm cr}>0$, the power of $T_{\mathrm{cr}}$ is
$\text{Power}_{\rm cr}   \approx \Phi \left( - z_{1-\alpha} + {\gamma_{\rm cr}}/{\sigma_{\rm cr}}\right),$ 
where $\Phi (\cdot)$ is the cumulative distribution function of the standard normal distribution and $\approx$ denotes asymptotic approximation.

Now suppose we only use the data at time 1, which is effectively a parallel-group design. The parallel-group counterpart to $\hat\theta_{\rm cr}$ is a simple mean difference $\hat{\theta}_{\mathrm{pr}}= \frac{1}{n_1} \sum_{i=1}^n A_i Y_{i1} -\frac{1}{n_0} \sum_{i=1}^n (1-A_i) Y_{i1}$. Under Assumption \ref{assump: random}, we have $E\big(\hat{\theta}_{\mathrm{pr}} \big) =\theta_1$ and $ \sqrt{n} \big( \hat{\theta}_{\mathrm{pr}} -\theta_1 \big) \xrightarrow{d} N(0, \sigma_{\mathrm{pr}}^2)$, where $\sigma_{\mathrm{pr}}^2= \pi_1^{-1}\var \big(Y_{i1}^{(1)} \big) +\pi_0^{-1}\var \big(Y_{i1}^{(0)} \big)$. For testing the null hypothesis $H_0:  \theta_1 = \theta^*$ versus the alternative hypothesis $H_A: \theta_1 > \theta^*$, the test statistic is $T_{\mathrm{pr}}= \sqrt{n} \big(\hat{\theta}_{\mathrm{pr}} -\theta^* \big) /\hat{\sigma}_{\mathrm{pr}}$, where $\hat{\sigma}_{\mathrm{pr}}^2$ is a consistent estimator for ${\sigma}_{\mathrm{pr}}^2$. Under the local alternative $\sqrt{n} (\theta_1 -\theta^*) \rightarrow \gamma_{\rm pr}$ with a constant $\gamma_{\rm pr}>0$, the power of $T_{\mathrm{pr}}$ is 
\begin{align} 
	\text{Power}_{\rm pr} \approx \Phi \left( -z_{1-\alpha} + \frac{\gamma_{\rm pr}}{\sigma_{\rm pr}} \right). \label{eq: power pr}
\end{align}

Let $\text{Power}_{\rm cr}= \text{Power}_{\rm pr}= 1-\beta$, we obtain the required sample sizes to achieve power of $1-\beta$ under the two designs: 
\begin{align*}
	n_{\rm cr}= \frac{(z_{1-\alpha} +z_{1-\beta})^2 \sigma_{\rm cr}^2}{ \{2^{-1} (\theta_1+\theta_2)- \theta^*\}^2} \text{ and } n_{\rm pr}= \frac{(z_{1-\alpha} +z_{1-\beta})^2 \sigma_{\rm pr}^2}{ (\theta_1- \theta^*)^2}.
\end{align*}
These sample size formulas have been previously derived, for example, in \citet[Section 2.4]{jones2014design} under additional assumptions such as the time-invariant treatment effect. 
When $\theta_1= \theta_2$, the two aforementioned null hypotheses are the same, and the ratio of the sample sizes to achieve the same power using the two tests, which is also called the Pitman asymptotic relative efficiency, is 
\[
\frac{n_{\rm cr}}{n_{\rm pr}}
=\frac{\sigma_{\rm cr}^2}{ \sigma_{\rm pr}^2} \cdot \frac{(\theta_1 - \theta^*)^2}{  \{2^{-1} (\theta_1+\theta_2)- \theta^*\}^2}.
\]
For illustration, consider a simple case where $\theta_1=\theta_2$,
$\var \big(Y_{it}^{(j)} \big)= \sigma^2$ for $t=1,2$ and $j=0,1$, and $\cov \big(Y_{i1}^{(j)}, Y_{i2}^{(1-j)} \big) ={\rho} \sigma^2$, where ${\rho \in [0,1)}$ is the intraclass correlation coefficient (ICC). Then $\sigma_{\rm cr}^2= 2(1-\rho) \sigma^2$ and $\sigma_{\rm pr}^2= 4\sigma^2$, and $n_{\rm cr}/ n_{\rm pr}=(1-\rho)/2 $. 
This implies that the crossover design only requires $(1-\rho)/2$ {i.e., at most half} of the sample size required by the parallel-group design, which is the key advantage of the crossover design over the parallel-group design.

\vspace{-5mm}

\section{Crossover trials with carry-over effect}
\label{sec: crossover}
\subsection{Setup and assumptions}
When there exists the carry-over effect in a crossover trial, the treatment at time 1 may interfere with the outcome at time 2;  see Figure \ref{fig: all assump}(b) for an illustration of the six potential outcomes with carry-over effects. In this case, Assumption \ref{assump: random} still holds by the act of randomization, while Assumption \ref{assump: no carry over} is violated.

We present a way of parameterizing the expected values of the six potential outcomes in Table \ref{tb: parameterization} that directly extends the parameterization used in Section \ref{sec: no carry-over}. In Table \ref{tb: parameterization}, $\tilde \theta_2= E\big(Y_{i2}^{(11)} -Y_{i2}^{(00)} \big)$ denotes the average treatment effect at time $2$, which compares the potential outcome had one stayed on treatment 1 to the potential outcome had one stayed on treatment 0; $\tilde \tau = E\big(Y_{i2}^{(00)} - Y_{i1}^{(0)} \big)$ denotes the expected {temporal} change in outcome had one stayed on treatment 0; $\lambda_{0}= E\big( Y_{i2}^{(10)}- Y_{i2}^{(00)} \big)$ and $\lambda_{1}= E\big( Y_{i2}^{(11)} - Y_{i2}^{(01)} \big)$ are the expected carry-over effects under the two treatment regimes.

\begin{table}[h]
	\caption{Parameterization of the six potential outcome means with and without Assumption \ref{assump: no carry over}. The two shadowed rows correspond to the potential outcomes that are never observed.
	} \label{tb: parameterization}
	\centering
	\begin{tabular}{|ccc|} \toprule
		Potential outcome means & With Assumption \ref{assump: no carry over} & Without Assumption \ref{assump: no carry over} \\ 
		$E( Y_{i1}^{(0)} )$ & $\mu$ & $\mu$\\ 
		$E( Y_{i1}^{(1)} )$ & $\mu+\theta_1$ & $\mu+\theta_1$ \\
		\rowcolor{gray!20} $E( Y_{i2}^{(00)} )$& $\mu+\tau$ &  $\mu+\tilde\tau$ \\
		$E( Y_{i2}^{(10)} )$& $\mu+\tau$ & $\mu+\tilde\tau+\lambda_{0}$ \\
		\rowcolor{gray!20} $E( Y_{i2}^{(11)} )$& $\mu+\tau+\theta_2$ &  $\mu+\tilde\tau+\tilde\theta_2$\\
		$E( Y_{i2}^{(01)} )$& $\mu+\tau+\theta_2$ & $\mu+\tilde\tau+\tilde\theta_2 - \lambda_{1}$\\\bottomrule
	\end{tabular}
\end{table}

\subsection{The basic estimator}

Theorem \ref{theo: basic estimator violate} derives 
the statistical properties of the basic estimator $\hat\theta_{\rm cr}$ defined in \eqref{eq: basic estimator} without Assumption \ref{assump: no carry over}. 

\begin{theorem}\label{theo: basic estimator violate}
	Under Assumption 1, \\
	(a) $E\big( \hat \theta_{\rm cr} \big) = 2^{-1} \big(\theta_1 + \tilde\theta_2 - \lambda_{0} - \lambda_{1} \big)$.  \\
	(b) $\sqrt{n} \big\{\hat \theta_{\rm cr} - 2^{-1} \big(\theta_1 + \tilde\theta_2 - \lambda_{0} - \lambda_{1} \big) \big\} \xrightarrow{d} N\big( 0,  \tilde \sigma^{2}_{\rm cr}  \big)$, where $ \tilde\sigma^{2}_{\rm cr}  = (4\pi_1)^{-1} \var \big(Y_{i1}^{(1)}-Y_{i2}^{(10)} \big)  +(4\pi_0)^{-1} \var \big(Y_{i2}^{(01)}-Y_{i1}^{(0)} \big)$.
\end{theorem}

Similar to the proof of Theorem \ref{theo: basic estimator}(a), we show that the expected change in outcome for $A_i=1$ is  $E(Y_{i1} - Y_{i2} \mid A_i=1)= E\big(Y_{i1}^{(1)} - Y_{i2}^{(10)} \big)= E\big(Y_{i1}^{(1)} -Y_{i1}^{(0)} + Y_{i1}^{(0)} - Y_{i2}^{(00)} + Y_{i2}^{(00)} - Y_{i2}^{(10)} \big) = \theta_1 - \tilde \tau - \lambda_{0}$, that for $A_i=0$ is $E( Y_{i2} - Y_{i1} \mid A_i=0)= E\big( Y_{i2}^{(01)} - Y_{i1}^{(0)} )= E(Y_{i2}^{(01)} - Y_{i2}^{(11)} + Y_{i2}^{(11)} - Y_{i2}^{(00)} + Y_{i2}^{(00)}-Y_{i1}^{(0)} \big) = \tilde \theta_2 + \tilde \tau - \lambda_{1}$. Hence, the overall average change in outcome is $2^{-1} \big(\theta_1 + \tilde\theta_2 - \lambda_{0} - \lambda_{1} \big)$. Hence, Theorem \ref{theo: basic estimator violate}(a) implies that $\hat\theta_{\rm cr}$ is still {an unbiased estimator} of {the} treatment effect $2^{-1} \big(\theta_1 + \tilde\theta_2 \big)$ under a population-level no carry-over effect assumption ($E\big(Y_{i2}^{(1k)} \big)= E \big(Y_{i2}^{(0k)} \big)$ for $k=0,1$), which is weaker than the individual-level no carry-over effect assumption stated in Assumption \ref{assump: no carry over}.  It is also straightforward to verify that Theorem \ref{theo: basic estimator} is a special case of Theorem \ref{theo: basic estimator violate} under Assumption \ref{assump: no carry over}.
Lastly, note that $\hat{\sigma}_{\rm cr}^2$ defined in (\ref{eq: basic est var}) is still a consistent estimator of $ {\tilde\sigma_{\rm cr}}^2$ without Assumption \ref{assump: no carry over}.

\subsection{Type I error and power analysis}
\label{subsec: type I error and power with carry over}

With possible carry-over effects, the first question one would ask is whether using $T_{\rm cr}$ to analyze crossover trials would lead to type I error inflation. To answer this question, consider the null hypothesis $H_0: 2^{-1} \big(\theta_1 +\tilde \theta_2 \big)= \theta^*$ versus the alternative hypothesis $H_A:2^{-1} \big(\theta_1 +\tilde \theta_2 \big)> \theta^*$ for some pre-specified $\theta^*$.
The type I error rate of $T_{\rm cr}= \sqrt{n}\big(\hat{\theta}_{\rm cr} -\theta^* \big)/ \hat{\sigma}_{\rm cr}$ is 
\begin{align}
	P_{H_0} (T_{\rm cr} >z_{1-\alpha}) & = P_{H_0} \left( \sqrt{n} \cdot \frac{\hat{\theta}_{\rm cr} -E \big(\hat{\theta}_{\rm cr} \big) +E\big(\hat{\theta}_{\rm cr} \big)  -\theta^*}{\hat{\sigma}_{\rm cr}} >z_{1-\alpha} \right)\nonumber \\
	&
	\approx \Phi \left( -z_{1-\alpha} -\sqrt{n} \cdot \frac{ 2^{-1} (\lambda_{0}+\lambda_{1})}{ \tilde \sigma_{\rm cr}} \right). \nonumber
\end{align}
When $\lambda_{0}+\lambda_{1}=0$, $\hat \theta_{\rm cr}$ is an unbiased estimator for the treatment effect $2^{-1} \big(\theta_1 +\tilde \theta_2 \big)$ and the type I error rate of  $T_{\rm cr}$ is $\alpha$. When $\lambda_{0}+\lambda_{1}>0$,  $\hat \theta_{\rm cr}$ underestimates the treatment effect $2^{-1} \big(\theta_1 +\tilde \theta_2 \big)$ and the type I error rate is less than $\alpha$, meaning that the test is conservative. When $\lambda_{0}+\lambda_{1}<0$, $\hat \theta_{\rm cr}$ overestimates the treatment effect $2^{-1} \big(\theta_1 +\tilde \theta_2 \big)$ and the type I error rate is larger than $\alpha$, meaning that the test is invalid. Therefore, even with the carry-over effect, the crossover design can still control the type I error rate of a one-sided test when $\lambda_{0}+\lambda_{1}\geq 0$. 

However, a carry-over effect that does not inflate the type I error rate can have a negative impact on the power. Specifically, with possible carry-over effects and under the local alternative $\sqrt{n} \big\{ 2^{-1} (\theta_1 +\tilde \theta_2)- \theta^* \big\} \rightarrow \gamma_{\rm cr}$ with a constant $\gamma_{\rm cr}>0$, the power of $T_{\rm cr}$ is 
\begin{align}
	{\text{Power}}_{\rm cr} \approx \Phi \left( -z_{1-\alpha}+  \frac{ \gamma_{\rm cr}- \sqrt{n} 2^{-1} (\lambda_{0}+\lambda_{1}) }{ \tilde \sigma_{\rm cr}} \right). \label{eq: power cr general}
\end{align}
{\citet[Table 1]{senn1997} briefly notes the impact of carry-over on power but does not provide a power formula. Our Equation \eqref{eq: power cr general} provides a general power formula in the presence of carry-over effects, which produces results that are the same as those in \citet[Table 1]{senn1997}
	(with $n=44$, $\lambda_0+\lambda_1 =0, 0.5, 1, ..., 5, \tilde\sigma_{\rm cr}^2= 96, 2^{-1}(\theta_1+\tilde\theta_2) = 5, \theta^*=0, \alpha=0.025$).} Equation \eqref{eq: power cr general} also implies that the required sample size to achieve $1-\beta$ power is 
\begin{align*}
	\tilde n_{\rm cr}= \frac{(z_{1-\alpha} + z_{1-\beta})^2 \tilde \sigma_{cr}^2 }{ \{2^{-1}(\theta_1 + \tilde\theta_2 - \lambda_{0} - \lambda_{1} )-\theta^*\}^2}. 
\end{align*}

For illustration, consider a test of superiority and suppose that $\theta_1 =\tilde\theta_2= \theta_{\mathrm{Alt}} >0$ and the Pitman asymptotic relative efficiency between the crossover and the parallel-group design is 
\begin{align}
	\frac{\tilde{n}_{\rm cr}}{n_{\rm pr}} = \frac{\tilde\sigma_{\rm cr}^2 }{\sigma_{\rm pr}^2 } \cdot \left\{ 1 - \frac{\lambda_{0}+\lambda_{1}}{2\ \theta_{\rm Alt}}\right\}^{-2}. \nonumber
\end{align}
When $\theta_1 + \tilde\theta_2 >  \lambda_{0} +\lambda_{1} \geq 0$, i.e., the carry-over effects are non-negative {for the purpose of controlling type I error rate} and are smaller than the treatment effects, the difference between $n_{\rm pr}$ and $\tilde n_{\rm cr}$ equals
\[
\tilde n_{\rm cr} - n_{\rm pr} = c \left\{   \frac{2^{-1} (\lambda_{0}+\lambda_{1} )}{ \theta_{\rm Alt}} - \bigg(1 - \frac{\tilde\sigma_{\rm cr}}{\sigma_{\rm pr}} \bigg) \right\},
\]
where $c$ is some positive constant. Therefore, in order for the crossover design to have higher efficiency than the parallel-group one, we need the carry-over effect to be small. Specifically, similar to our discussion at the end of Section \ref{sec: no carry-over}: when $\tilde \sigma_{\rm cr}/ \sigma_{\rm pr} = \sqrt{(1-\rho)/2}$, we need 
$\frac{2^{-1} (\lambda_{0}+\lambda_{1} ) }{ \theta_{\rm Alt}} < 1-  \sqrt{(1-\rho)/2}$.
When $\rho= 0.3,0.5,0.7$, we need to have $2^{-1} (\lambda_0+ \lambda_1)>0$ and less than $0.41 \theta_{\rm Alt}, 0.5\theta_{\rm Alt}, 0.61 \theta_{\rm Alt}$ respectively so that the type I error is not inflated and the crossover design is more powerful than the parallel-group one. This small carry-over effect condition may be plausible in many scenarios because carry-over effects are usually relatively small compared to the treatment effects. Therefore, a crossover design can still be more powerful than a parallel-group design in many cases even when the carry-over effect exists. In Section S2 of the supplementary materials, we discuss how to control the type I error when $\lambda_{0}+\lambda_{1}< 0$ using a sensitivity analysis approach \citep{rosenbaum2020design}.



\section{Covariate adjustment for crossover trials}
\label{sec: covariate adjustment}

Adjusting for prognostic baseline covariates in the analysis of randomized controlled trials is encouraged by regulatory agencies because it has high potential to improve efficiency under approximately the same minimal statistical assumptions that would be needed for unadjusted estimation \citep{fda:2019aa}. It often uses a \textit{working} model between the outcomes and covariates, but its estimand is the same as when using the unadjusted method and its inference does not rely on the working model being correctly specified. 

Covariate adjustment for parallel-group trials has been extensively studied recently. {From the theory of semiparametrics \citep{robins1994estimation, tsiatis2006semiparametric}, \cite{tsiatis2008covariate}  established a general class of consistent and asymptotically normal estimators for the average treatment effect. When linear working models are used, covariate adjustment using}   an analysis of heterogeneous covariance (ANHECOVA) working model that includes all treatment-by-covariate interaction terms can lead to \textit{guaranteed efficiency gain} regardless of the model is misspecified or not \citep{yang2001efficiency, Lin:2013aa, ye2022inference, ye2022toward}. These recent results have not been extended to crossover trials, although covariate adjustment is broadly recommended for crossover trials \citep{metcalfe2010analysis, mehrotra2014recommended, jemielita2016improved}. 

We consider adjusting for a baseline covariate vector  $\bfX_i$ measured before randomization (i.e., pre-randomization covariates) using the following covariate-adjusted ANHECOVA estimator
\begin{align}
	\hat \theta_{\rm cr, adj}=\frac{\big\{ \bar \Delta_1 - \hat\bbeta_1^T (\bar \bfX_1 - \bar \bfX) \big\} - \big\{ \bar \Delta_0 - \hat\bbeta_0^T (\bar \bfX_0 - \bar \bfX) \big\} }{2}  , \nonumber
\end{align}
where $\bar \bfX$ is the sample mean of all $\bfX_i$'s,  $\bar \bfX_a$ is the sample mean of $\bfX_i$'s from subjects with $A_i=a$,  $\Delta_i$ and $\bar\Delta_a$ are defined in \eqref{eq: basic estimator},
and $ \hat\bbeta_a = \left\{\sum_{i: A_i = a} (\bfX_i - \bar \bfX_a) (\bfX_i - \bar \bfX_a)^T   \right\}^{-1} \sum_{i: A_i = a} (\bfX_i - \bar \bfX_a)  \Delta_i $
is the least squares estimator of $\bbeta_a$ from fitting the linear working model $E(\Delta_i \mid A_i=a, \bfX_i) = \mu_a +  \bbeta_a^T \bfX_i   $ using subjects with $A_i=a$.

The following heuristics reveal why ANHECOVA does not change the estimand, often gains but never hurts efficiency even when the linear working model is wrong. As randomization balances the covariate distribution, both $\bar \bfX_a $ and $\bar \bfX$ estimate the same quantity and thus, $\hat\bbeta_a^T (\bar \bfX_a - \bar \bfX)$ is an ``estimator'' of zero. Hence, $\bar \Delta_a - \hat\bbeta_a^T (\bar \bfX_a - \bar \bfX)  $ and $\bar \Delta_a$  correspond to the same estimand. In addition, as $n\rightarrow \infty$, $\hat\bbeta_a$ converges to $\bbeta_a= \var(\bfX_i)^{-1} \cov (\bfX_i, \Delta_i\mid A_i=a )$ in probability, regardless of the linear working model is correct or not. Hence, $\bar \Delta_a - \hat\bbeta_a^T (\bar \bfX_a - \bar \bfX)  $ is asymptotically equivalent to $\bar \Delta_a - \bbeta_a^T (\bar \bfX_a - \bar \bfX)  $, whose variance is
\begin{align*}
	\var\{ \bar \Delta_a - \bbeta_a^T (\bar \bfX_a - \bar \bfX) \} &=    \var( \bar \Delta_a ) + \var \{ \bbeta_a^T (\bar \bfX_a - \bar \bfX)  \} - 2 {\rm cov} \{\bar \Delta_a ,  \bbeta_a^T (\bar \bfX_a - \bar \bfX) \} \\
	&=  \var( \bar \Delta_a ) - \var \{ \bbeta_a^T (\bar \bfX_a - \bar \bfX)  \}. 
\end{align*}
Consequently, the asymptotic variance of $\bar \Delta_a - \hat\bbeta_a^T (\bar \bfX_a - \bar \bfX)  $ is no larger than that of $\bar \Delta_a$. These results are formally stated in Theorem \ref{theo: adj estimator}, which are all in the asymptotic sense.

\begin{theorem}\label{theo: adj estimator}
	Under Assumption 1,\\
	(a)
	$	\sqrt{n} \big\{	\hat \theta_{\rm cr, adj} - 2^{-1}(\theta_1 + \tilde\theta_2 - \lambda_{0} - \lambda_{1} )  \} \xrightarrow{d} N\left( 0, \tilde \sigma^{2}_{\rm cr, adj}  \right)$, where $\tilde\sigma^{2}_{\rm cr, adj}  =  (4\pi_1)^{-1} \var \big(Y_{i1}^{(1)}-Y_{i2}^{(10)} - \bbeta_1^T \bfX_i \big)  +(4\pi_0)^{-1} \var \big(Y_{i1}^{(0)} - Y_{i2}^{(01)} - \bbeta_0^T \bfX_i \big)  + 4^{-1} (\bbeta_1 - \bbeta_0 )^T \var(\bfX) (\bbeta_1 - \bbeta_0 ) $, and $\bbeta_a= \var(\bfX_i)^{-1} \cov (\bfX_i, \Delta_i\mid A_i=a )$. \\
	(b) Moreover, $\tilde \sigma^{2}_{\rm cr, adj}  -\tilde\sigma^{2}_{\rm cr} = - (4\pi_1\pi_0)^{-1} (\pi_0\bbeta_1+ \pi_1 \bbeta_0)^T \var(\bfX) (\pi_0\bbeta_1+ \pi_1 \bbeta_0)\leq 0$. 
\end{theorem} 

Theorem \ref{theo: adj estimator} is proved by applying Theorem 1 and Corollary 1 in \cite{ye2022toward} with $\Delta_i$ as the outcome.  From Theorem \ref{theo: adj estimator}(b), we see that the asymptotic variance of $\hat \theta_{\rm cr, adj}$ is no larger than that of $\hat \theta_{\rm cr}$, where the equality holds if and only if $\pi_0\bbeta_1+ \pi_1 \bbeta_0=0$.  
This occurs, for example, when $\bbeta_0=\bbeta_1=0$, i.e., when the covariates are uncorrelated with the change in outcome $\Delta_i$. In fact, Theorem 1 of \cite{ye2022toward} implies a stronger result that $	\hat \theta_{\rm cr, adj}$ has the smallest asymptotic variance among all linearly adjusted estimators of the form $2^{-1} \big[\big\{ \bar \Delta_1 - \boldsymbol{b}_1^T (\bar \bfX_1 - \bar \bfX) \big\} - \big\{ \bar \Delta_0 - \boldsymbol{b}_0^T (\bar \bfX_0 - \bar \bfX) \big\} \big]$, where $\boldsymbol{b}_0, \boldsymbol{b}_1$ are any fixed or random vectors that have the same dimension as $\bfX_i$. A consistent estimator of  $\tilde \sigma^{2}_{\rm cr, adj}$ is 
\begin{align*}
	\hat \sigma_{\rm cr, adj}^2 = \frac{1}{4\pi_1} S_{\Delta1, {\rm adj}}^2 + \frac{1}{4\pi_0} S_{\Delta0, {\rm adj}}^2 + \frac14 \big(\hat\bbeta_1- \hat\bbeta_0 \big)^T \hat\Sigma_X \big(\hat\bbeta_1- \hat\bbeta_0 \big),
\end{align*}
where $S_{\Delta a, {\rm adj}}^2$ is the sample variance of $Y_{i1}- Y_{i2} - \hat\bbeta_a^T \bfX_i$ based on subjects under $A_i=a$, for $a=0,1$, and $\hat \Sigma_X$ is the sample covariance matrix of $\bfX_i$ based on the entire sample. One can easily construct a $Z$-test based on the covariate-adjusted estimator  $\hat \theta_{\rm cr, adj} $, which from Theorem \ref{theo: adj estimator} is guaranteed to be more powerful than the unadjusted counterpart $T_{\rm cr}$ in the asymptotic sense.

{In this article, we have focused on adjusting for pre-randomization covariates. As discussed above, if the covariates are related to the change in outcome $\Delta_i$, which could happen, for example, when there exist treatment-by-covariates interactions, then adjusting for pre-randomization covariates can reduce the variability of $\Delta_i$ and improve asymptotic efficiency. This is also shown in our simulation study in Section 5. However, in crossover trials, this adjustment typically yields small to modest benefits unless the covariates strongly influence the change in outcome. Previous studies have explored adjusting for period-dependent baseline covariates, i.e., covariates that are measured before treatment is given within each period \citep{kenward2010use, jones2014design}. Yet, one important caveat is that the washout period needs to be sufficiently long to ensure that the period-dependent baselines are not influenced by previous treatment (i.e., by carry-over). {In instances where the carry-over effect is a concern, more sophisticated methods to adjust for period-dependent
		covariates should be used (e.g. the longitudinal G methods \citep[Chapter 21]{hernanrobins}). Exploring these methods will be a direction for future research.
}}

\vspace{-5mm}
\section{Power calculations}
\label{sec: power}
To compare the power of parallel-group and crossover design, we consider two situations. In Case I, we consider the following simple data-generating process for which we can calculate the power using both the formula and simulations:
\begin{align*}
	&Y_{i1}^{(0)}= X_{i1}+ X_{i2}+ X_{i3} + \epsilon_{i1}\\
	&Y_{i1}^{(1)}= \theta_1+ X_{i1}+ X_{i2}+  X_{i3} + \epsilon_{i2}\\
	&Y_{i2}^{(10)}= \tilde \tau+ \lambda_{0}+ X_{i1}+ X_{i2}+ b  X_{i3} +\epsilon_{i3} \\
	&Y_{i2}^{(01)}= \tilde\tau + \tilde \theta_2  - \lambda_{1} + X_{i1}+ X_{i2}+ b  X_{i3}+ \epsilon_{i4} .
\end{align*}
In {Case II}, we consider a more complex data-generating process of the potential outcomes with treatment-by-covariate interactions:
\begin{align*}
	&Y_{i1}^{(0)}= X_{i1}+X_{i2}+X_{i3}  + \epsilon_{i1}\\
	&Y_{i1}^{(1)}=X_{i1}+ X_{i2}+ 2X_{i3}+ \theta_1 (0.5+ I_{(X_{i2}>0)}+ X_{i1}X_{i3}) + \epsilon_{i2}\\ 
	&Y_{i2}^{(10)}= X_{i1}+ X_{i2}+ X_{i3}+ \tilde\tau  (1+X_{i1}) + \lambda_0 (1+ X_{i2}X_{i3})+ \epsilon_{i3}\\
	&Y_{i2}^{(01)}= X_{i1}+ X_{i2}+ 2X_{i3}+ \tilde\tau (1+ X_{i1})+ \tilde\theta_2 ( 1+ 2X_{i1}X_{i3} )  -\lambda_1 (1+ X_{i2}X_{i3}) +  \epsilon_{i4}. 
\end{align*}
In both cases, $X_{ij},\epsilon_{ik}  \sim N(0,1)$ for $j=1,2,3$ and $k=1,2,3,4$.  The treatment arm indicator $A_i$ is Bernoulli with $\pi_0=\pi_1=1/2$. The observed outcomes for subject $i$ are $Y_{i1} = Y_{i1}^{(1)}, Y_{i2} = Y_{i2}^{(10)}$ if $A_i=1$, and $Y_{i1} = Y_{i1}^{(0)}, Y_{i2} = Y_{i2}^{(01)}$ if $A_i=0$. The observed data are $(\bfX_i, A_i, Y_{i1}, Y_{i2}), i=1,\dots, n$.  For the covariate-adjusted estimator, we adjust for $\bfX_i = ( X_{i1}, X_{i2}, X_{i3} )^T$. We set  $n=500, \tilde \tau=0.2, \alpha=0.025$, $\theta^*=0$, and $b\in \{0,\frac13\}$.

In Case I, because of the simple data-generating process, we can directly calculate that $\var \big( Y_{i1}^{(0)} \big)= \var \big( Y_{i1}^{(1)} \big)=4$,
$\var \big( Y_{i2}^{(10)} \big)= \var \big( Y_{i2}^{(01)} \big)= b^2 +3$, 
$\mathrm{Cov} \big( Y_{i1}^{(1)}, Y_{i2}^{(10)} \big)$ $=\mathrm{Cov} \big( Y_{i1}^{(0)}, Y_{i2}^{(01)} \big)= b+2 $, and thus $\sigma_{\rm pr}^2 = 16$ and $\tilde \sigma_{\mathrm{cr}}^2= (1-b)^2 +2$.  In this case, $\var(\bfX_i)$ is a $3\times3$ identity matrix and $\bbeta_0=\bbeta_1 =( 0, 0 ,1-b )^T$. Thus, $\tilde \sigma_{\mathrm{cr,adj}}^2= 2$. From Corollary 1 in \cite{ye2022toward}, it is also easy to calculate the asymptotic variance of the covariate-adjusted ANHECOVA estimator using only the parallel-group design as  $\sigma_{\rm pr, adj}^2 = 4$. Hence, we calculate the power based on both formula and simulations. Figure \ref{fig: power formula case1} shows the type I error rate and power for four one-sided tests $T_{\rm pr}, T_{\rm pr,adj}, T_{\rm cr}, T_{\rm cr, adj}$ under Case I when $\theta_1= \tilde\theta_2= \theta \in [0,0.5]$ and $\lambda_0=\lambda_1=\lambda\in \{-0.1, 0, 0.1, 0.3\}$ based on the formula in \eqref{eq: power pr}, \eqref{eq: power cr general} and its covariate-adjusted counterpart. Note that $T_{\rm pr, adj}$ is based on the ANHECOVA estimator that takes the same form as $\hat\theta_{\mathrm{cr, adj}}$ but with $\Delta_i$ replaced by $Y_{i1}$.  In this setting, the value of $b$ only affects the power of $T_{\rm cr}$ but not the other two tests, so we present $b=\frac13$ for $T_{\rm cr}$ only. In the supplementary materials, we present the type I error rate and power obtained by simulation, which is shown to agree with the power in Figure \ref{fig: power formula case1}.

\begin{figure}[h]
	\centering
	\includegraphics[width=\textwidth]{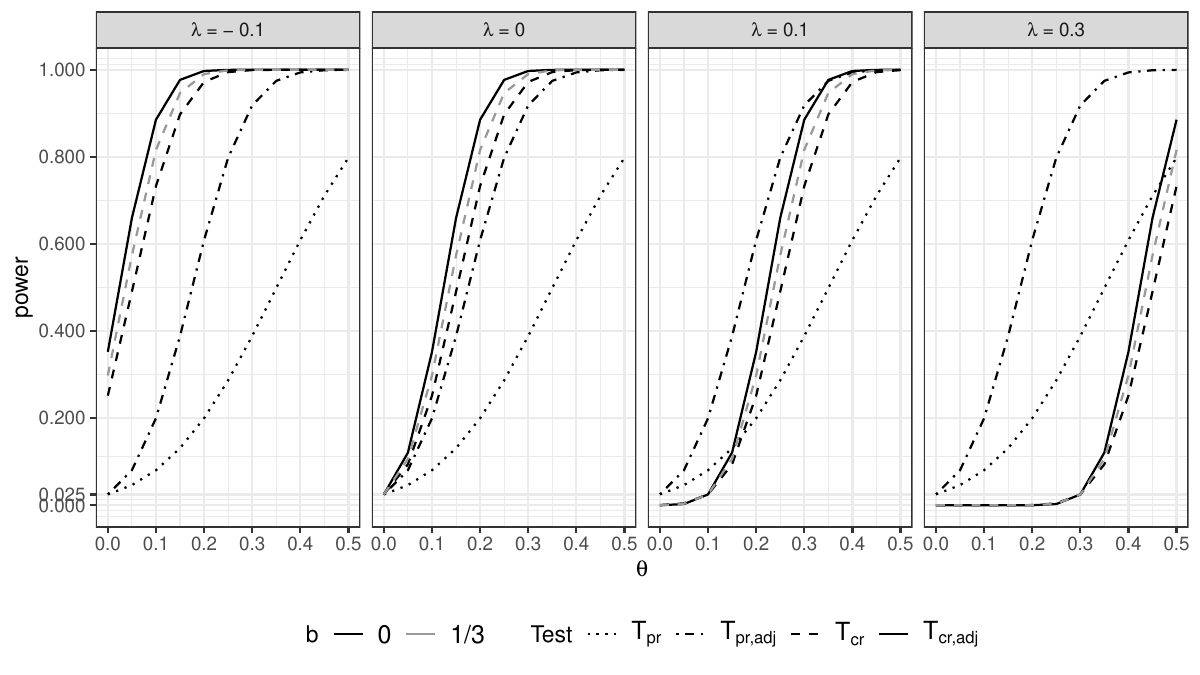}
	\caption{Power curves for four tests under Case I calculated by formula when $\lambda =-0.1, 0, 0.1, 0.3$ and $b=0, \frac13$ under Case I.
		Note that the power of $T_{\rm pr}$, $T_{\rm pr,adj}$, and $T_{\rm cr, adj}$ are unaffected by $b$.}\label{fig: power formula case1}
\end{figure}

In Figure \ref{fig: power formula case1}, the power of $T_{\rm pr}$ and $T_{\rm pr, adj}$ are not affected by $\lambda$ and can be used to benchmark the performance of the other tests. When $\theta=0$, i.e., under the null hypothesis of no treatment effect, the type I error rates of $T_{\rm pr}$ and $T_{\rm pr,adj}$ are equal to $\alpha=0.025$. The type I error rates of $T_{\rm cr}$ and $T_{\rm cr, adj}$ are greater than $\alpha$ when $\lambda<0$, equal to  $\alpha$ when $\lambda=0$, and smaller than $\alpha$ when $\lambda>0$. In other words, $T_{\rm cr}$ and $T_{\rm cr, adj}$ can control the type I error rate when $\lambda\geq 0$; meanwhile, $T_{\rm cr}$ and $T_{\rm cr, adj}$ become more and more conservative as $\lambda$ grows. 

When $\theta>0$, i.e., under the alternative hypothesis, the power comparison of $T_{\rm cr}$ and $T_{\rm pr}$ depends on the sign of $\lambda - \theta\left(1- \sqrt{\frac{2-b}{8}}\right)$ following the calculations at the end of Section \ref{subsec: type I error and power with carry over}. Specifically,  with $b=0$,  $T_{\rm  cr}$ is more powerful than $T_{\rm pr}$ when $\lambda<0.5 \theta$. Hence, with no carry-over effect ($\lambda=0$), we see that $T_{\rm cr}$ has substantial power gain compared to $T_{\rm pr}$; with a small carry-over effect ($\lambda=0.1$),  $T_{\rm cr}$ is more powerful than $T_{\rm pr}$ when $\theta>0.2$; with a large carry-over effect ($\lambda=0.3$), $T_{\rm pr}$ is more powerful than $T_{\rm cr}$ for $0<\theta\leq 0.5$. Furthermore, since a larger $b$ leads to a larger ICC, it also results in a slightly higher power of $T_{\rm cr}$. Lastly, adjusting for covariates that are related to the change in outcome can always increase power regardless of whether the model is correct or not. Hence, we see that the covariate-adjusted test $T_{\rm cr, adj}$ is always more powerful than the unadjusted test $T_{\rm cr}$; in some cases, the power gain can be up to 17\%.

Figure \ref{fig: power case2} shows the empirical type I error rate and power for Case II obtained by simulation from 10,000 repetitions. The results are similar to the results from Figure \ref{fig: power case2}, and covariate adjustment can improve the power when there exist treatment-by-covariates interactions.

\begin{figure}[h]
	\centering
	\includegraphics[width=\textwidth]{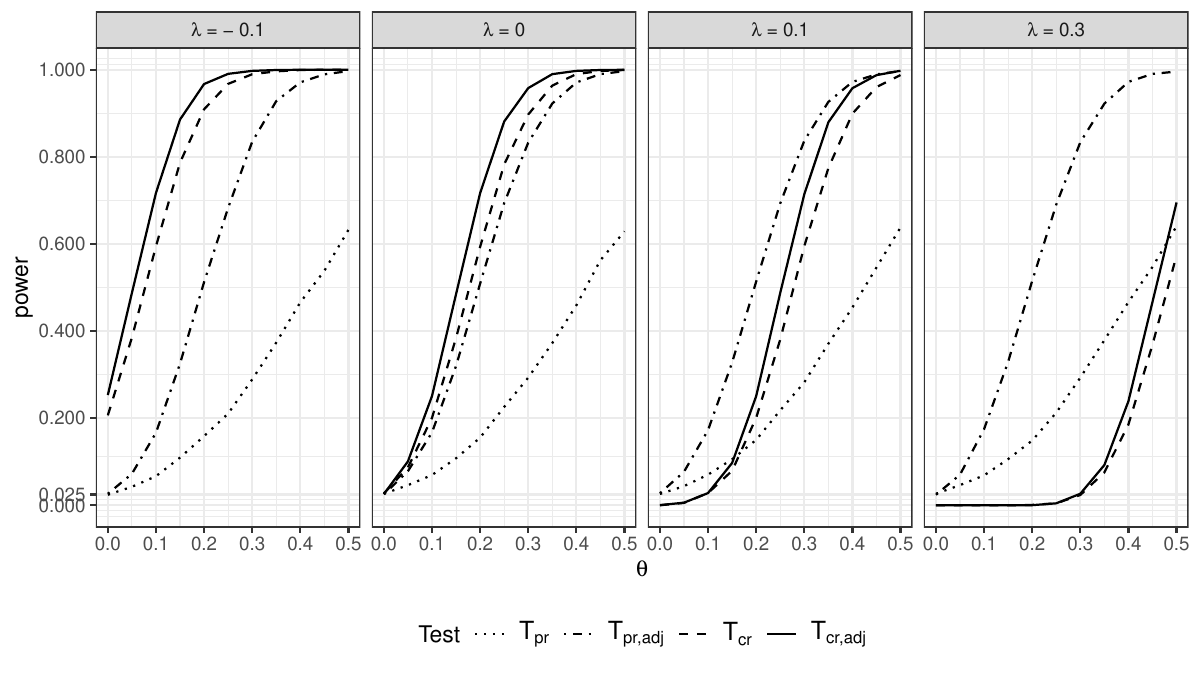}
	\caption{Power curves for four tests under Case II calculated by simulations when $\lambda =-0.1, 0, 0.1, 0.3$.}\label{fig: power case2}
\end{figure}


\vspace{-5mm}
\section{Real Data Example}
\label{sec: mtn034}
\subsection{Application to the REACH study}
In this section, we revisit the REACH study \citep{nair2023adherence} described in Section \ref{sec: intro} and apply our methods to study the adherence difference between monthly  DVR and daily oral TDF/FTC.  

In the REACH study, a total of 247 participants were enrolled and randomized (1:1) to two treatment arms of product use: using the DVR for 6 months and switching to daily oral TDF/FTC for a second 6 months, or using the daily oral TDF/FTC for 6 months and then switching to DVR for 6 months. Participants’ baseline characteristics are summarized in  Table \ref{tb:MTN034}.

\begin{table}[h]
	\caption{Baseline characteristics by randomized crossover sequence in the REACH study. Data are presented as median (IQR) or $n$ (\%).  
	} \label{tb:MTN034}
	\centering
	\begin{tabular}{lccc}
		\toprule
		& \textbf{DVR first}   & \textbf{TDF/FTC first}   & \textbf{Overall}      \\
		& $(n=124)$ &$(n=123)$ & (n=247) \\ \hline 
		Age in years                        \\
		\quad Median (IQR)    & 18 (17-19)                                                                          & 18 (17-19)                                                                          & 18 (17-19)   \\
		Site location                                       &                                                                                      &                                                                                      &              \\
		\quad South Africa - Cape Town                            & 30 (24.2\%)                                                                          & 30 (24.4\%)                                                                          & 60 (24.3\%)  \\
		\quad South Africa - Johannesburg                         & 34 (27.4\%)                                                                          & 33 (26.8\%)                                                                          & 67 (27.1\%)  \\
		\quad Uganda - Kampala                                    & 30 (24.2\%)                                                                          & 30 (24.4\%)                                                                          & 60 (24.3\%)  \\
		\quad Zimbabwe - Harare                                   & 30 (24.2\%)                                                                          & 30 (24.4\%)                                                                          & 60 (24.3\%)  \\
		Marital status                                      &                                                                                      &                                                                                      &              \\
		\quad Single                                              & 106 (85.5\%)                                                                         & 108 (87.8\%)                                                                         & 214 (86.6\%) \\
		\quad Married/cohabiting                                  & 16 (12.9\%)                                                                          & 14 (11.4\%)                                                                          & 30 (12.1\%)  \\
		\quad Separated/divorced                                  & 2 (1.6\%)                                                                            & 1 (0.8\%)                                                                            & 3 (1.2\%)    \\
		Highest level of education                    &                                                                                      &                                                                                      &              \\
		\quad Primary                                             & 15 (12.3\%)                                                                          & 18 (14.6\%)                                                                          & 33 (13.5\%)  \\
		\quad Secondary                                           & 92 (75.4\%)                                                                          & 97 (78.9\%)                                                                          & 189 (77.1\%) \\
		\quad Higher                                              & 15 (12.3\%)                                                                          & 8 (6.5\%)                                                                            & 23 (9.4\%)   \\
		Earns own income                                    & 22 (17.7\%)                                                                          & 31 (25.2\%)                                                                          & 53 (21.5\%)  \\
		Ever pregnant                              & 52 (41.9\%)                                                                          & 47 (38.2\%)                                                                          & 99 (40.1\%)  \\
		Currently has a sex partner                         & 111 (89.5\%)                                                                         & 108 (87.8\%)                                                                         & 219 (88.7\%) \\
		Worry about HIV infection                           & 86 (39.4\%)                                                                          & 82 (66.7\%)                                                                          & 168 (68.0\%) \\
		Diagnosed with STI & 40 (32.3\%)                                                                          & 47 (38.2\%)                                                                          & 87 (35.2\%)  \\
		\quad Syphilis                                            & 3 (2.4\%)                                                                            & 3 (2.4\%)                                                                            & 6 (2.4\%)    \\
		\quad Trichomoniasis                                      & 8 (6.5\%)                                                                            & 5 (4.1\%)                                                                            & 13 (5.3\%)   \\
		\quad Gonorrhea                                           & 13 (10.5\%)                                                                          & 8 (6.5\%)                                                                            & 21 (8.5\%)   \\
		\quad Chlamydia                                           & 30 (24.2\%)                                                                          & 41 (33.3\%)                                                                          & 71 (28.7\%)  \\
		CES depression scale                                          &                                                                                      &                                                                                      &              \\
		\quad  $\geq 12$  (indicative of depression)       & 28 (22.6\%)                                                                          & 38 (30.9\%)                                                                          & 66 (26.7\%)  \\
		\quad $< 12$          & 89 (71.8\%)                                                                          & 78 (63.4\%)                                                                          & 167 (67.6\%) \\
		\quad Missing                                             & 7 (5.6\%)                                                                            & 7 (5.7\%)                                                                            & 14 (5.7\%)   \\
		Alcohol use disorder (AUDIT-C score)                             &                                                                                      &                                                                                      &              \\
		\quad $\geq 3$ (indicative of alcohol use disorder)              & 50 (40.3\%)                                                                          & 47 (38.2\%)                                                                          & 97 (39.3\%)  \\
		\quad $< 3$                      & 74 (59.7\%)                                                                          & 75 (61.0\%)                                                                          & 149 (60.3\%)  \\
		\quad Missing                                             & 0 (0.0\%)                                                                          & 1 (0.8\%)                                                                          & 1 (0.4\%) \\
		Preference for products                             &                                                                                      &                                                                                      &              \\
		\quad  DVR                                       & 51 (41.1\%)                                                                          & 43 (35.0\%)                                                                          & 94 (38.1\%)  \\
		\quad oral TDF/FTC                                     & 43 (34.7\%)                                                                          & 57 (46.3\%)                                                                          & 100 (40.5\%) \\
		\quad Equal preference                                    & 29 (23.4\%)                                                                          & 22 (17.9\%)                                                                          & 51 (20.6\%)  \\
		\quad Did not answer                                      & 1 (0.8\%)                                                                            & 1 (0.8\%)                                                                            & 2 (0.8\%)   \\
		\bottomrule
		\multicolumn{4}{l}{\small IQR=Interquartile range; STI=sexually transmitted infection; CES=Center for Epidemologic Studies;} \\
		\multicolumn{4}{l}{\small AUDIT-C= Alcohol Use Disorders Identification Test-Concise.}
	\end{tabular}
\end{table}

{We apply our method to compare adherence to monthly DVR and daily oral TDF/FTC. The binary adherence endpoint is defined as high use at the end of period 1 and period 2, with high use of daily oral TDF/FTC defined as  tenofovir-diphosphate concentrations greater
	than or equal to 700 fmol/punch (associated with taking an average of four or more tablets per week in the
	previous month), and high use of monthly  DVR  defined as greater than or equal to 4 mg dapivirine released from the returned ring (continuous use for
	28 days in the previous month).
	These adherence definitions are widely accepted in the literature \citep{nair2023adherence}. It is important to note, however, that the definitions require continuous daily usage of the ring, but not for oral TDF/FTC. For covariate adjustment, we consider baseline covariates that are likely to be associated with adherence, including age, site location, preference for either product, depression, alcohol use disorder, sexually transmitted infection, and participants’ level of worry about HIV infection. We use the single imputation method to handle the missing baseline covariates \citep{zhao2022adjust}. 
	
	{To discuss the impact of potential carry-over effects in the analysis of this trial, we first note that for the  DVR first arm, using unadjusted sample means, the estimated adherence rate to  DVR at period 1 was 55.6\% and to TDF/FTC was 44.4\% at period 2; for the TDF/FTC first arm, the estimated adherence rate to TDF/FTC at period 1 was 48.0\% and 
		to DVR was 52.8\% at period 2. Hence, the average adherence difference at period 1 is 55.6\%-48.0\%=7.6\%, and average adherence difference at period 2 is 52.8\%-44.4\%= 8.4\%. This small difference between the treatment effect at period 1 and period 2 may be attributed to a larger treatment effect in period 2, and/or a negative carry-over effect ($\lambda_1$) for the reason discussed in Section 1. However, because the difference is small, even if the carry-over effect exists, it is unlikely to substantially alter the result. 
	}

	Table \ref{tb: MTN034result} presents both the parallel and crossover estimates and their covariate-adjusted counterparts, alongside the standard errors and 95\% confidence intervals. The parallel estimates use data from period 1 only. Estimates from all methods are similar, indicating that adherence to  DVR is about 7.7\%-8.5\% higher compared to daily oral TDF/FTC.  The SEs of the crossover estimators are smaller compared to the SEs of the parallel estimators, demonstrating that the crossover design is more efficient. Adjusting for covariates leads to slightly larger estimates and smaller standard errors. {All the 95\% confidence intervals (except for the covariate-adjusted crossover one) cover 0, suggesting that the difference is not statistically significant.}

	\begin{table}[h]
		\caption{Covariates-adjusted and unadjusted parallel and crossover estimators for the average treatment effect of DVR on adherence compared to daily oral TDF/FTC, with standard errors and 95\% confidence intervals.} \label{tb: MTN034result}
		\centering
		\begin{tabular}{llll}
			\toprule
			Type & Mean & SE & 95\% confidence interval\\
			\hline 
			Parallel, unadjusted & 0.077 & 0.064 & $(-0.048, 0.202)$\\
			Parallel, adjusted & 0.082 & 0.062 & $(-0.040, 0.204)$\\
			Crossover, unadjusted & 0.081 & 0.043 & $(-0.003, 0.165)$\\
			Crossover, adjusted & 0.085 & 0.042 & $(0.003, 0.168)$\\
			\bottomrule
		\end{tabular}
	\end{table}

	\subsection{Using the REACH study to design a hypothetical trial}
	Despite being highly effective for HIV prevention, adherence to daily oral TDF/FTC is low among women \citep{celum2019prep}. The dual prevention pill (DPP), a daily oral pill combining oral contraceptives and TDF/FTC, has the potential to increase women's adherence to daily oral TDF/FTC \citep{friedland2021}. In this section, using data from the REACH study, we use simulations to evaluate the power trade-off between the crossover design and parallel-group design in a hypothetical trial comparing adherence to a single daily DPP versus adherence to two daily pills (one pill of TDF/FTC and one pill of oral contraceptive). The central hypothesis is that a DPP regime can increase women’s adherence to TDF/FTC.

	Suppose adolescent girls and young women with baseline characteristics similar to those in the REACH study are recruited for this trial, and we use the same design as described in Section 6.1 to compare DPP and two daily pills. In this example, we anticipate a nonnegative carry-over effect primarily because the DPP is expected to increase adherence in the first period, and those who adhere consistently in the first period tend to develop a routine or habit around taking daily oral pills, which may positively influence their adherence in the second period. From the results in Section 3, using a crossover design with the nonnegative carry-over effect does not lead to type I error inflation.

	Our simulations are based on the data of 123 participants who took TDF/FTC first in the REACH study. In the simulation study, we set $\pi_1=\pi_0= 0.5$, $\tilde\tau=-0.05$, $\theta_1= \tilde\theta_2= \theta$ where $\theta\in \{0.05, 0.08, 0.10, 0.15\}$ for different treatment effects. The ICC (i.e., $\rho$) is estimated to be $\rho=0.39$ from the 123 participants in the REACH study. Under the simple scenario described at the end of Section 3, when $\frac{ 2^{-1}(\lambda_0+ \lambda_1)}{ \theta}> 1-\sqrt{(1-0.39)/2} \approx 0.447$, the crossover design is less powerful than the parallel one. So here we set $\lambda_1= \lambda_0= \lambda \in\{0, 0.25\theta, 0.45\theta\} $ for different carry-over effects. We consider four test statistics: $T_{\rm pr}, T_{\rm pr, adj}, T_{\rm cr}, T_{\rm cr, adj}$.  
	The simulation process involves the following steps:
	\begin{enumerate}
		\item Denote the baseline covariates and the adherence variable at week 24 from the REACH study as $\big(\bfX_i, Y_{i1}^{(0)} \big), i=1\dots, N$, with $N=123$. The average of $Y_{i1}^{(0)}$ is 0.48.
		Fit a logistic regression model for the probability of adherence to TDF/FTC using age, site location, STI status, depression, alcohol use disorder, worry about HIV infection and baseline preferences for the products. Single imputation is applied for all missing values as described in Section 6.1. 
		Denote the fitted model as $\hat\mu_1(\boldsymbol{x}) = \mathrm{expit} \big( \hat \alpha_{0}+ \hat \bbeta_1^T \boldsymbol{x} \big)$, where $\mathrm{expit} (x)= \frac{\exp(x)}{1+\exp(x)}$.
		
		\item Generate $Y_{i1}^{(1)}\sim \mathrm{Bernoulli} \big(\mathrm{expit} \big(\hat \alpha_{1}+ \hat \bbeta_1^T \bfX_i \big) \big)$, where  $\hat\alpha_1$ satisfies $n^{-1}\sum_{i=1}^n \mathrm{expit} (\hat \alpha_{1}+ \hat \bbeta_1^T \bfX_i) - n^{-1}\sum_{i=1}^n \mathrm{expit}  \big(\hat \alpha_{0}+ \hat \bbeta_1^T \bfX_i\big) = \theta$.
		
		\item Generate $Y_{i2}^{(01)}$ from $Y_{i2}^{(01)} |Y_{i1}^{(0)} =1 \sim \text{Bernoulli} \left( \frac{s_0}{0.48} \right)$ and $ Y_{i2}^{(01)} |Y_{i1}^{(0)} =0 \sim \text{Bernoulli} \left( \frac{0.48+ \tilde\tau +\theta -\lambda -s_0}{ 1-0.48} \right)$, where $s_0$ is chosen (as a function of $\tilde\tau, \theta, \lambda$)  to ensure that  $\corr (Y_{i2}^{(01)}, Y_{i1}^{(0)})= \rho$ and $E\big(Y_{i2}^{(01)})- E(Y_{i1}^{(0)}\big)= \tilde\tau+ \theta -\lambda$; see details in Section S3.3 of the supplementary materials.
		
		\item Generate $Y_{i2}^{(10)}$ from $Y_{i2}^{(10)}| Y_{i1}^{(1)}=1 \sim \text{Bernoulli} \left( \frac{s_1}{0.48+\theta} \right)$ and $Y_{i2}^{(10)}| Y_{i1}^{(1)}=0 \sim  \text{Bernoulli}\left( \frac{0.48+ \tilde\tau+ \lambda -s_1}{1-0.48-\theta} \right)$, where $s_1$ is chosen  (as a function of $\tilde\tau, \theta, \lambda$) to ensure that  $E(Y_{i2}^{(10)})= 0.48 + \tilde\tau +\lambda $ and $\corr(Y_{i2}^{(10)}, Y_{i1}^{(1)})= \rho$.
		By now, we have obtained the covariates and 4 potential outcomes $\big( \bfX_i, Y_{i1}^{(0)}, Y_{i1}^{(1)}, Y_{i2}^{(01)}, Y_{i2}^{(10)} \big)$ for 123 individuals.
		
		\item A random sample of size $n$ is drawn from the 123 individuals' covariates and potential outcomes with replacement. Each of the obtained $n$ subjects is assigned randomly to $A_i=0$ or 1 with equal probability. For $A_{i}=1$, $ Y_{i1} = Y_{i1}^{(1)}$ and $ Y_{i2}= Y_{i2}^{(10)}$; for $A_{i}=0$, $ Y_{i1} = Y_{i1}^{(0)}$ and $Y_{i2}= Y_{i2}^{(01)}$. Hence, we obtain the observed data $\left( A_i, \bfX_i, Y_{i1}, Y_{i2}\right), i=1,\dots, n$, and calculate the test statistics $T_{\rm pr}, T_{\rm pr, adj}, T_{\rm cr}, T_{\rm cr, adj}$.
		
		\item Steps (2)-(5) are repeated 2,000 times for each $\theta$ and $\lambda$, from which we obtain the empirical powers of the test statistics. 
	\end{enumerate}
	
	\begin{table}[h]
		\centering
		\caption{Empirical power of $T_{\rm pr}, T_{\rm pr, adj}, T_{\rm cr}, T_{\rm cr, adj}$ for a hypothetical trial using data from the REACH study.}
		\begin{tabular}{lllcccc}
			\toprule
			$n$ & $\lambda$ & $\theta$ & $\text{Power}_\mathrm{pr}$ & $\text{Power}_\mathrm{pr, adj}$ & $\text{Power}_\mathrm{cr}$ & $\text{Power}_\mathrm{cr, adj}$ \\
			\hline 
			325 & $0 $ & 0.05 &  0.214 &0.236  & 0.423 &0.446 \\
			&  & 0.08 &  0.342 &0.379 &  0.670 &0.681\\
			&  & 0.10 & 0.439 & 0.484 &  \textbf{0.797} &0.818\\
			&  & 0.15 & 0.738 & 0.770  & 0.973 &0.978 \\
			\hline 
			680 & $0.25 \theta$ & 0.05 & 0.324 &0.354   &0.452 &0.463 \\
			&  & 0.08 & 0.534 &0.570  & 0.680 &0.6701\\
			&  & 0.10 &0.679 &0.720 &  \textbf{0.799} &0.808\\
			&  & 0.15 & 0.905 &0.923 &  0.959 &0.961\\
			\hline 
			1900 & $0.45 \theta$ & 0.05 & 0.558 & 0.579  & 0.523 &0.528 \\
			&  & 0.08 & 0.778 &0.789  & 0.684 &0.688\\
			&  & 0.10 & 0.876 &0.889 &  \textbf{0.791} &0.797\\
			&  & 0.15 &0.988 &0.989 &  0.929 &0.935\\
			\bottomrule
		\end{tabular}
		\label{tb:simu results}
	\end{table}
	
	Table \ref{tb:simu results} shows the empirical powers for $T_{\rm pr}, T_{\rm pr, adj}, T_{\rm cr}, T_{\rm cr, adj}$ for different $\theta$ and $\lambda$ based on 2,000 simulations. The sample sizes $n=325, 680, 1900$ are determined such that $\text{Power}_{\text{cr}}$ is close to 80\% when $\theta=0.10$. From Table \ref{tb:simu results}, first, we can see that the empirical power increases when $\theta$ increases. Second, under either the parallel or crossover design, covariate adjustment always leads to a larger power. Third, when $\lambda=0$, the crossover tests have considerably larger power than the parallel-group tests. When $\lambda= 0.25 \theta$, the crossover tests still have larger power, but the power difference is not as pronounced. When $\lambda=0.45\theta$, the crossover tests become slightly less powerful than the parallel-group tests.

	\section{Discussion}
	\label{sec: discussion}
	The crossover is an efficient trial design that uses participants as their own controls. The carry-over effect, especially the behavioral carry-over effect, is an outstanding concern because it can bias the estimation of the treatment effect. Using a potential outcomes framework and minimal statistical assumptions, we investigate the impact of the carry-over effect in a two-treatment two-period crossover trial. Our results provide a clear characterization of how the carry-over effect influences estimation bias, as well as the type I error and power of tests. When the carry-over effect $\lambda_1+\lambda_0$ is negative, the basic estimator overestimates the treatment effect, which can inflate the type I error of one-sided tests. Conversely, when the carry-over effect $\lambda_1+\lambda_0$ is positive, the basic estimator underestimates the treatment effect, which does not inflate the type I error of one-sided tests but negatively impacts the power. Furthermore, when $\lambda_1+\lambda_0$ is positive but {relatively }small {compared} to the treatment effect, the crossover design can still be more powerful than the parallel-group design. We can further apply covariate adjustment in crossover trials for guaranteed efficiency gain. All the methods in this article can be implemented using the \textsf{RobinCar}  package in \textsf{R}  \citep{ye_2023_8319889}.

	We discuss our results in the context of two studies and demonstrate that the anticipated carry-over effect can be either non-positive or non-negative: one is the REACH study evaluating monthly DVR and daily oral (TDF/FTC), the other is a hypothetical future trial evaluating a single daily DPP and two daily pills. Through these two example studies, our key message is that while the crossover design is an efficient design, the carry-over effect needs to be carefully considered in trial planning and execution. 
	
	In this article, our primary focus has been on the power trade-off between the parallel-group design and the crossover design, particularly in the context of carry-over effects. 
	However, this is just one among various practical considerations. First, a crossover trial typically requires a longer follow-up time, about twice as long as a parallel-group design, and this duration may be even longer if a washout period is incorporated to reduce carry-over effects \citep{senn2002time}. The longer follow-up time can increase costs and the risk of dropout. Second, the longer commitment and the need to switch treatments in a crossover design might discourage participants, potentially slowing recruitment. On the other hand, the opportunity to try multiple treatments in one study can make the crossover design more appealing to participants in certain scenarios \citep{diener2019headache}. Meanwhile, in some trials, experiencing multiple treatments is essential for the trial objective. For example, in the REACH study, the participants need to experience both treatments before entering the choice period, which allows for the study of their preference between the two treatments.



	\section*{Acknowledgements}
	We are grateful to the study participants, study staff and investigators on the MTN-034/REACH study who provided the data for this analysis. We would also like to thank the anonymous referee, an Associate Editor, the Editor, and Professor Elizabeth Brown for their constructive comments. This work was supported by National Institute of Allergy and Infectious Diseases [NIAID 5 UM1 AI068617].
	
	\vspace{-5mm}
	\section*{Supplementary Materials}
	
	Web Appendices referenced in Sections 2-6 are available with this paper at the Biometrics website on Oxford Academic. Code necessary to reproduce simulation and application 
	results can be found in the supplementary materials.\vspace*{-8pt}

	\section*{Data Availability}
	The data that support the findings in this paper are available from the Microbicide Trials Network. Restrictions apply to the availability of these data, which were used under license for this study. Data are available from the authors with the permission of the Microbicide Trials Network.

\bibliographystyle{apalike}
\bibliography{reference}

\clearpage
\begin{center}
{\sffamily\bfseries\LARGE
Supplementary Materials
}
\end{center}
\setcounter{equation}{0}
\setcounter{table}{0}
\setcounter{lemma}{1}
\setcounter{section}{0}
\renewcommand{\theequation}{S\arabic{equation}}
\renewcommand{\thetable}{S\arabic{table}}
\renewcommand{\thefigure}{S\arabic{figure}}
\renewcommand{\thesection}{S\arabic{section}}

\section{Additional Power Comparison}
In this section, we check our power calculation by the formula in Section 5 with simulation powers using 10,000 repetitions. Figure \ref{fig: power formula, all} is simply adding the simulation power results on top of Figure 2, and we find good agreement between the powers by formula and powers by simulation. 
\begin{figure}[h] \centering
	\includegraphics[width=\textwidth]{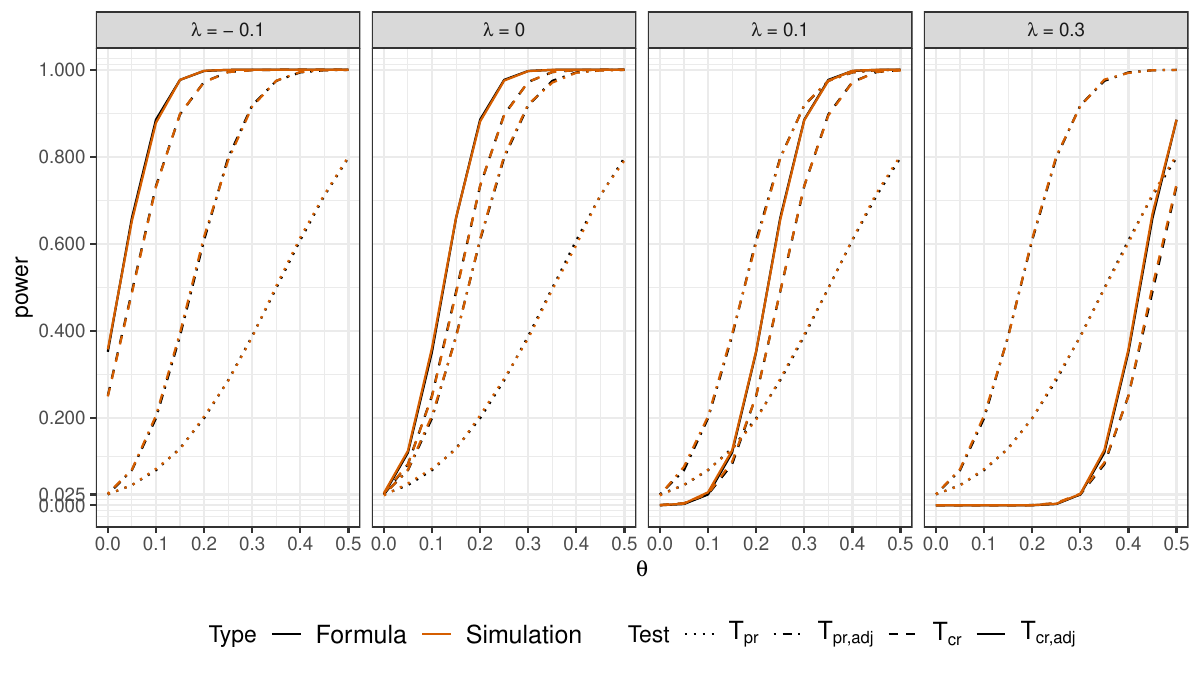}
	\caption{Power curves for four tests under Case I calculated by formula and simulations (with 10,000 repetitions) when $\lambda =-0.1, 0, 0.1, 0.3$ and $b=0$. The obtained power curves using these two methods agree with each other.}\label{fig: power formula, all}
\end{figure}

\section{Controlling for type I error when $\lambda_0 + \lambda_1<0$}

Similar to \cite{yi2022testing}, we consider a sensitivity parameter $\Lambda < 0 $ that bounds the bias, i.e., $\frac12 (\lambda_0 +\lambda_1) \geq  \Lambda $. Then the rejection region 
\begin{align*}
	\frac{\sqrt{n} \left(\hat \theta_{\rm cr} - \theta^* + \Lambda \right)}{\hat\sigma_{\rm cr}} > z_{1-\alpha}
\end{align*}
can control the type I error at level $\alpha$. 

For the implementation of the sensitivity analysis, practitioners are not required to specify the value of the sensitivity parameter $\Lambda$. Results from the sensitivity parameter can be summarized by the ``tipping point'' -- the magnitude of $\Lambda$ that would be needed such that the null hypothesis can no longer be rejected  \citep{Cornfield1959, rosenbaum2020design}. If such a value of $\Lambda$ is deemed implausible, then we still have evidence to reject the null hypothesis.

\section{Technical Proofs}
\subsection{Proof of Theorem 1}
\begin{proof}
	(a)	Let $\tau= E\left(Y_{i2}^{(0)} - Y_{i1}^{(0)} \right)$ denote the expected time trend in the absence of treatment. Then
	\begin{align*}
		&	E(\hat\theta_{\rm cr} )  \\
		&= \frac{1}{2} E\left\{ \frac{1}{n_1} \sum_{i=1}^{n} A_i  ( Y_{i1} - Y_{i2} )  +  \frac{1}{n_0} \sum_{i=1}^{n} (1-A_i)  ( Y_{i2} - Y_{i1} )   \right\} \\
		&= \frac{1}{2} E\left\{ \frac{1}{n_1} \sum_{i=1}^{n} A_i  \left( Y_{i1}^{(1)} - Y_{i2}^{(0)} \right)  +  \frac{1}{n_0} \sum_{i=1}^{n} (1-A_i)  \left( Y_{i2}^{(1)} - Y_{i1}^{(0)} \right)   \right\} \\
		&= \frac{1}{2} E\left\{ \frac{1}{n_1} \sum_{i=1}^{n} A_i  \left( Y_{i1}^{(1)} -Y_{i1}^{(0)} +   Y_{i1}^{(0)} - Y_{i2}^{(0)} \right)  +  \frac{1}{n_0} \sum_{i=1}^{n} (1-A_i)  \left( Y_{i2}^{(1)} - Y_{i2}^{(0)} + Y_{i2}^{(0)}- Y_{i1}^{(0)} \right)   \right\} \\
		&= \frac{1}{2} (\theta_1 - \tau) + \frac12 (\theta_2 + \tau) = \frac{\theta_1+\theta_2}{2}.
	\end{align*}
	
	(b)  Asymptotic normality is straightforward from the central limit theorem. In what follows, we derive the asymptotic variance of $\sqrt{n}(\hat\theta_{\rm cr}  -\frac{\theta_1+\theta_2}{2})$. Note that 
	\begin{align*}
		\hat\theta_{\rm cr} -\frac{\theta_1+\theta_2}{2} =&  \underbrace{\frac{1}{2} \left\{ \frac{1}{n_1} \sum_{i=1}^n A_i \left( Y_{i1}^{(1)}- Y_{i2}^{(10)} -\theta_1+\tau \right) + \frac{1}{n_0} \sum_{i=1}^n (1-A_i) \left( Y_{i2}^{(01)}- Y_{i1}^{(0)} -\theta_2-\tau \right) \right\}}_{M_1}
	\end{align*}
	We can show that 
	\begin{align*}
		&\sqrt{n} M_1 \mid A_1,\dots, A_n \xrightarrow{d} N\left(0, \frac{\var \left(Y_{i1}^{(1)}-Y_{i2}^{(0)} \right) }{4\pi_1} + \frac{\var \left(Y_{i2}^{(1)}-Y_{i1}^{(0)} \right)}{ 4\pi_0} \right). 
	\end{align*}
	From the bounded convergence theorem, this result still holds unconditionally, i.e., 
	\begin{align*}
		\sqrt{n} \left(\hat\theta_{\rm cr}  - \frac{1}{2} (\theta_1 +\theta_2) \right) \xrightarrow{d} N\left(0, \frac{\var \left(Y_{i1}^{(1)}-Y_{i2}^{(0)} \right) }{4\pi_1} + \frac{\var \left(Y_{i2}^{(1)}-Y_{i1}^{(0)} \right)}{ 4\pi_0} \right).
	\end{align*}
	
\end{proof}

\subsection{Proof of Theorem  2}
(a) We calculate the expectation as follows:
\begin{align*}
	E(\hat\theta_{\rm cr}) &= \frac{1}{2} E \left\{ \frac{1}{n_1} \sum_{i=1}^n A_i \left( Y_{i1}^{(1)}- Y_{i2}^{(10)} \right)+ \frac{1}{n_0} \sum_{i=1}^n (1-A_i) \left( Y_{i2}^{(01)}- Y_{i1}^{(0)} \right) \right\}\\
	&= \frac{1}{2} E \bigg\{ \frac{1}{n_1} \sum_{i=1}^n A_i \left( Y_{i1}^{(1)}- Y_{i1}^{(0)}+ Y_{i1}^{(0)}- Y_{i2}^{(00)}+ Y_{i2}^{(00)}- Y_{i2}^{(10)} \right)\\
	&\quad + \frac{1}{n_0} \sum_{i=1}^n (1-A_i) \left( Y_{i2}^{(01)}- Y_{i2}^{(11)}+ Y_{i2}^{(11)}- Y_{i2}^{(00)}+ Y_{i2}^{(00)}- Y_{i1}^{(0)} \right)\\
	&= \frac{1}{2} (\theta_1- \tilde\tau -\lambda_0)+ \frac{1}{2} \left(\tilde{\theta}_2 -\lambda_1 +\tilde{\tau} \right) =\frac{1}{2} \left( \theta_1 +\tilde{\theta}_2 -\lambda_0 -\lambda_1 \right).
\end{align*}

(b)  Asymptotic normality is straightforward from the central limit theorem. In what follows, we derive the asymptotic variance of $\sqrt{n}(\hat\theta_{\rm cr}  - 2^{-1} \left( \theta_1+ \tilde{\theta}_2 -\lambda_0 -\lambda_1) \right)$. Similar to the proof of Theorem 1(b), we have that 
\begin{align*}
	\sqrt{n} \left(\hat\theta_{\rm cr}  - \frac{1}{2} \left( \theta_1 +\tilde{\theta}_2 -\lambda_0 -\lambda_1 \right) \right) \xrightarrow{d} N\left(0, \frac{\var(Y_{i1}^{(1)}-Y_{i2}^{(0)}) }{4\pi_1} + \frac{\var(Y_{i2}^{(1)}-Y_{i1}^{(0)})}{ 4\pi_0} \right).
\end{align*}

\subsection{Generating correlated binary random variables}
\begin{lemma}\label{theo: correlated binary}
	Consider a binary variable $Z_1 \sim \mathrm{Bernoulli} (p_1)$ where $\mathrm{Pr} (Z_1=1) =p_1$. Given some $p_2 \in (0,1)$ and
	\begin{align*}
		{-1}& {\le \max \left\{ -\sqrt{ \frac{p_1 p_2 }{(1-p_1) (1-p_2)}}, -\sqrt{ \frac{(1-p_1) (1-p_2) }{p_1 p_2}} \right\} < \rho}\\ 
		&{<\min \left\{ \sqrt{ \frac{p_1 (1-p_2) }{p_2 (1-p_1)}}, \sqrt{ \frac{p_2 (1-p_1) }{p_1 (1-p_2)}} \right\} \le1.}
	\end{align*}
	If $Z_2|Z_1=1 \sim \mathrm{Bernoulli} \left( \frac{s}{p_1} \right)$ and $Z_2|Z_1=0 \sim \mathrm{Bernoulli} \left( \frac{p_2-s}{1-p_1} \right)$ where $s=\rho \sqrt{p_1 (1-p_1) p_2 (1-p_2)}$ $+p_1 p_2$, then $E(Z_2) =p_2$ and $\corr (Z_1, Z_2)= \rho$. 
\end{lemma}

\begin{proof}
	{First, it is straightforward to show that when $\rho$ satisfies the aforementioned condition, it implies that}
	$${\begin{cases} -1\le \rho \le 1\\ 0<\frac{s}{p_1}<1\\ 0< \frac{p_2-s}{1-p_1} <1 \end{cases} }$$
	Then, we could also show that
	\begin{align*}
		&\mathrm{Pr} (Z_2=1| Z_1=1)= \frac{ \mathrm{Pr} (Z_2=1, Z_1=1) }{\mathrm{Pr} (Z_1=1)}= \frac{s}{p_1} \Rightarrow \mathrm{Pr} (Z_2=1, Z_1=1)=s\\
		&\mathrm{Pr} (Z_2=1| Z_1=0)= \frac{ \mathrm{Pr} (Z_2=1, Z_1=0) }{\mathrm{Pr} (Z_1=0)}= \frac{p_2-s}{ 1-p_1} \Rightarrow \mathrm{Pr} (Z_2=1, Z_1=0)= p_2-s.
	\end{align*}
	This implies that 
	\begin{align*}
		& \mathrm{Pr} (Z_2=1)= \mathrm{Pr} (Z_2=1, Z_1=1)+ \mathrm{Pr} (Z_2=1, Z_1=0)= p_2= E(Z_2)
	\end{align*}
	Then we have $\var (Z_2)= p_2 (1-p_2)$. Given that $\var (Z_1)= p_1 (1-p_1)$, we have
	\begin{align*}
		\corr (Z_1, Z_2) &= \frac{\cov (Z_1,Z_2) }{ \sqrt{ \var (Z_1) \var (Z_2)}}= \frac{E(Z_1 Z_2) -p_1 p_2}{ \sqrt{p_1 (1-p_1) p_2 (1-p_2)}}\\
		&= \frac{\mathrm{Pr} (Z_1=1, Z_2=1) -p_1 p_2}{ \sqrt{p_1 (1-p_1) p_2 (1-p_2)}}
		= \frac{s- p_1 p_2}{ \sqrt{p_1 (1-p_1) p_2 (1-p_2)}} =\rho
	\end{align*}
\end{proof}


\end{document}